\documentclass[12pt,twoside]{report}
\usepackage{amsmath,amsfonts,amssymb,amsthm,mathabx}
\usepackage{times}
\usepackage{pdfsync}
\usepackage{cite}
\usepackage{url}
\usepackage{hyperref}
\usepackage{tensor}
\usepackage{color}
\usepackage[a4paper,top=25mm,bottom=30mm,right=48mm,left=30mm]{geometry}

\renewcommand{\a}{\alpha}
\renewcommand{\b}{\ensuremath{\beta}}
\renewcommand{\c}{\gamma}
\renewcommand{\d}{\delta}
\newcommand{\e}{\varepsilon}
\newtheorem{example}{Example}
\newtheorem{remark}{Remark}
\newtheorem{lemma}{Lemma}
\newtheorem{theorem}{Theorem}

\numberwithin{equation}{chapter} \makeatletter
\@addtoreset{equation}{chpater}

\usepackage{fancyhdr}
\usepackage{lipsum}

\begin{document}

\title{Introduction to Classical Gauge Field
  Theory and to Batalin-Vilkovisky Quantization}

\author{Glenn Barnich and Fabrizio Del Monte}

\date{}

\def\mytitle{Introduction to Classical Gauge Field
  Theory and to Batalin-Vilkovisky Quantization}

\addtolength{\headsep}{4pt}

\begin{centering}

  \vspace{1cm}

  \textbf{\Large{\mytitle}}

  \vspace{1.5cm}

  \textbf{Lectures held at the 22nd ``Saalburg'' Summer School (2016)} 

  \vspace{1.5cm}

  A course by
\vspace{.2cm}

{\Large 
Glenn~Barnich$^{\, a,}$\footnote{gbarnich@ulb.ac.be}
}

\vspace{.5cm}
Notes written and edited by

\vspace{.2cm}
{\Large 
Fabrizio Del Monte$^{\, b,}$\footnote{fabrizio.delmonte@sissa.it} 
}
\vspace{1cm}
{\small
\begin{center}
 {a) {\it Physique Th\'eorique et Math\'ematique}\\ 
{\it Universit\'e libre de Bruxelles \& International Solvay Institutes}\\
  {\it Boulevard du Triomphe, Campus Plaine C.P. 231, B-1050 Bruxelles, Belgique}} \\
  \vspace{1cm}
 {b) {\it International School for Advanced Studies (SISSA), Theoretical Particle Physics}\\
  {\it Via Bonomea, 265, 34136 Trieste, Italy and INFN, Sezione di Trieste}} \\
\end{center}
}
\end{centering}

\vfill
\thispagestyle{empty}
\newpage

{\footnotesize \tableofcontents}

\chapter{Locality and the horizontal complex}

This set of lectures loosely follows \cite{Henneaux:1992ig} by
emphasizing Lagrangian field theoretic aspects that are not covered
there in detail, and are not usually treated in standard textbooks
either. The choice of material is to a large extent
idiosyncratic. References to well-known original work underlying the
subject, such as
\cite{Becchi:1975nq,Zinn-Justin:1974mc,Batalin:1981jr,Voronov:1982ph}, 
have generally been omitted. They can be found in the cited reviews.

The goal of the first lecture is to rephrase the basic objects of
classical field theory in the language of jet bundles. Although this
requires a bit of work in the beginning, it is worthwhile because
objects such as actions and functional derivatives can be phrased in a
purely algebraic way in this framework.

Original work in this context goes back for instance to
\cite{Vinogradov:1977,Vinogradov:1978,Vinogradov:1984}. We follow here
the exposition of \cite{Andersonbook,Anderson1991,Olver:1993} (see
also \cite{Dickey:1991xa}), where detailed proofs and more references
to original work can be found. The first three chapters have been
treated in the current form in \cite{Barnich2001}.

\section{Jet-bundles}

In classical mechanics the central object is an \textbf{action
  functional} $S$, which is the integral of a Lagrangian density $L$:
\begin{equation}
S[q]=\int^{t_2}_{t_1} dt\, L\left(q^i(t),\dot{q}^i(t),t \right).
\end{equation}
The equations of motion are obtained by requiring the variation of the
action to be an extremum at a classical solution. More precisely, when
considering virtual variations at fixed time that vanish at the end
points, this amounts to requiring that the \textbf{Euler-Lagrange
  derivatives} of the Lagrangian vanish. If the Lagrangian $L$ is a
function of the generalized coordinates $q^i$ and their first order
time derivatives $\dot{q}^i$, the latter are defined by
\begin{align}\label{eq:EL1}
  \frac{\delta L}{\delta q^i}=\frac{\partial L}{\partial q^i}-\frac{d}{dt}
  \frac{\partial L}{\partial\dot{q}^i}, &&
 \frac{d}{dt}=\frac{\partial}{\partial t}
 +\dot{q}^i\frac{\partial}{\partial q^i}+\ddot{q}^i\frac{\partial}{\partial\dot{q}^i}.
\end{align}
\begin{remark}
  In order for the Euler-Lagrange operator to make sense, $q^i$ and
  $\dot{q}^i$ have to be considered as independent variables. The
  space on which they are coordinates is the (first order) jet-space.
\end{remark}
\begin{remark}
  When acting on a function that depends on the first order time
  derivatives, the total time derivative derivative $\frac{d}{dt}$
  involves second order time derivatives. In other words, it maps
  functions defined on the first order jet-space to functions defined
  on the second order jet-space.
\end{remark}

The underlying geometrical structure is described by introducing a new
mathematical object, the \textbf{jet bundle}. This is a fiber
bundle\footnote{The reader unfamiliar with the concept of fiber
  bundles should not be worried: for what concerns our discussion, it
  suffices to have in mind the conceptual idea that underlies them. As
  a Riemannian manifold is a smooth deformation of $\mathbb{R}^n$,
  which locally looks like flat space, a fiber bundle can be thought
  as a smooth deformation of a direct product of two spaces, that
  locally looks like the direct product itself.} which has as base
space $M$ the real line $\mathbb{R}$ in the case of classical mechanics
and Minkowski space $\mathbb{R}^{3,1} $ in the case of relativistic
field theory, and as fiber $V$ a manifold whose local coordinates are the
fields and their derivatives. Coordinates on the base are often called
``independent variables'', while coordinates on the fiber are
dependent ones.

In the case of classical mechanics, restricting to the
case of Lagrangians depending only on the coordinates and velocities,
the jet bundle will be parametrized by $t,q,\dot{q} $.

For a field theory, let $V$ be locally parametrized by the fields
$\phi^i$ (which we take to be even for simplicity), and let
\begin{equation}
E\equiv M\times V.
\end{equation}
Let us then denote by $V^k$ the \textbf{k-th jet space}, i.e. the
manifold obtained by extending the set of dependent coordinates to
\begin{equation}
\phi^i,\phi^i_\mu,\dots,\phi^i_{\mu_1\dots\mu_k},
\end{equation}
where the $\phi^i_{\mu_1\dots\mu_k}$ are completely symmetric in their
lower indices.

Locally, the \textbf{jet bundle}
$J^k(E)$ is the direct product
\begin{equation}
J^k(E)=M\times V^k.
\end{equation}
When taking suitable \textbf{sections} $s$, the coordinates of $V^k$
are to be identified with field ``histories'' and their derivatives up
to order $k$,
\begin{equation}
\begin{split}
s: \	& M \longrightarrow J^k(E) \\
  	& x^\mu\mapsto  (x^\mu,\phi^i(x),\frac{\partial\phi^i(x)}{\partial x^\mu},\dots).
\end{split}
\end{equation}

\begin{remark}
  The components $g_{\mu\nu}$ of the metric in General Relativity are
  not independent local coordinates for the fiber of the jet bundle,
  since they must satisfy, besides symmetry, the additional condition
\begin{equation}
\det g_{\mu\nu}\ne 0.
\end{equation}
\end{remark}

\vfill
\pagebreak

\textbf{Local functions}
\begin{equation}
f[x,\phi]=f(x,\phi^i,\dots,\phi^i_{\mu_1\dots\mu_k})
\end{equation}
are defined to be smooth functions on $J^k(E)$, for some (unspecified)
$k$. They are functions that depend on the base space coordinates
$x^\mu$, the fields and a finite number of their derivatives. Typical
examples of local functions are the Lagrangian (density) $L$ of
classical mechanics or of field theory.

A \textbf{vector field} on $J^k(E)$ has the form
\begin{equation}\label{eq:B1VecF}
  \textbf{v}=a^\mu\frac{\partial}{\partial x^\mu}
  +b^i\frac{\partial}{\partial\phi^i}
  +\sum_{l=1}^k \sum_{0\le\mu_1\le\dots\le\mu_k}b^i_{\mu_1\dots\mu_k}
  \frac{\partial}{\partial\phi^i_{\mu_1\dots\mu_l}}.
\end{equation}
This notation is not very practical, as it does not allow one to use
Einstein's summation convention for repeated indices. One thus defines
the \textbf{symmetrized derivative}
\begin{equation}
  \frac{\partial^s}{\partial\phi^i_{\mu_1\dots\mu_l}}=\frac{m_0!\dots
    m_{n-1}!}{l!}
  \frac{\partial}{\partial\phi^i_{\mu_1\dots\mu_l}},
\end{equation}
where $m_\mu$ denotes the number of times the index $\mu$ appears in
$\mu_1\dots\mu_l$, and in terms of which equation \eqref{eq:B1VecF}
takes the more compact form
\begin{equation}
  \textbf{v}=a^\mu\frac{\partial}{\partial x^\mu}
  +b^i\frac{\partial}{\partial\phi^i}
  +\sum_{l=1}^kb^i_{\mu_1\dots\mu_l}
  \frac{\partial^s}{\partial\phi_{\mu_1\dots\mu_l}}.
\end{equation}
This expression can be simplified further by using a \textbf{multi-index
  notation}:
\begin{equation}
(\mu)\equiv\{\emptyset,\mu_1,\mu_1\mu_2,\dots \},\quad
|\mu|=0,1,2,\dots, 
\end{equation}
with $(\mu)$ the multi-index and $|\mu|$ its length, the definition
of a vector field \eqref{eq:B1VecF} becomes
\begin{equation}
\textbf{v}=a^\mu\frac{\partial}{\partial x^\mu}+b^i_{(\mu)}\frac{\partial}{\partial\phi^i_{(\mu)}},
\end{equation}
in which Einstein's summation convention for multi-indices includes a
summation over their lengths. 

\section{Total and Euler-Lagrange derivatives}
\label{sec:total-euler-lagrange}

\textbf{Total derivatives}
are a generalization of $\frac{d}{dt} $ in equation \eqref{eq:EL1},
from classical mechanics to field theory, and to higher order
derivatives:
\begin{equation}
\begin{split}
  \partial_\nu & \equiv\frac{\partial}{\partial x^\nu}
  +\sum_{l}\phi^i_{\mu_1\dots\mu_l\nu}\frac{\partial^s}{\partial\phi^i_{\mu_1\dots\mu_l}} \\
  & =\frac{\partial}{\partial
    x^\nu}+\phi^i_{((\mu)\nu)}\frac{\partial}{\partial\phi^i_{(\mu)}}.
\end{split}
\end{equation}
Again, they are not vector fields on some $J^k(E)$ for a fixed $k$:
when acting on smooth functions on $J^k(E)$, they produce smooth
functions on $J^{k+1}(E)$. Total derivatives satisfy two important
properties: the first one is that, if a local function is evaluated at
a section,
\begin{equation}
f|_s=f\left(x,\phi^i(x),\dots,\frac{\partial\phi^i(x)}{\partial
    x^{\mu_1}\dots\partial x^{\mu_k}} \right), 
\end{equation}
then 
\begin{equation} \label{eq:1totalder}
\left(\partial_\mu f \right)|_s=\frac{d}{dx^\mu}\left(f|_s \right).
\end{equation}
The second one is that total derivatives commute:
\begin{equation}
[\partial_\mu,\partial_\nu]=0.
\end{equation}

The \textbf{Euler-Lagrange derivative}
of a local function $f$ is defined as
\begin{equation}\label{eq:1ELDer}
\begin{split}
  \frac{\delta f}{\delta\phi^i} &
  =\sum_{l=0}^k(-)^l\partial_{\mu_1}\dots\partial_{\mu_l}
  \frac{\partial^sf}{\partial\phi^i_{\mu_1\dots\mu_l}} \\
& =(-)^{|\mu|}\partial_{(\mu)}\frac{\partial f}{\partial\phi^i_{(\mu)}}.
\end{split}
\end{equation}
A first important lemma is the following:
\begin{lemma}\label{Lemma1}
  The Euler-Lagrange derivative of a local function is zero if and
  only if it is a total divergence, i.e.,
\begin{align}
\frac{\delta f}{\delta\phi^i}=0 && \Longleftrightarrow &&
f=\partial_\mu j^\mu,
       \label{equiv}                                                   
\end{align}
for some local functions $j^\mu$.
\end{lemma}

\begin{center}
  ----------
\end{center}

\subsection{Exercise 1}
\label{sec:exercice-1}

Prove Lemma 1.

\begin{proof}
Let us first prove that
\begin{align}
  f=\partial_\mu j^\mu && \Longrightarrow &&
  \frac{\delta f}{\delta\phi^i}=0.\label{condsuff}
\end{align}
In order to show this, we need to use the commutation relation
\begin{equation}
\left[\frac{\partial^s}{\partial\phi_{\mu_1\dots\mu_k}},\partial_\nu
\right]=\delta^{\mu_1}_{(\nu}
\delta^{\mu_2}_{\lambda_1} \dots\delta^{\mu_k}_{\lambda_{k-1})}
\frac{\partial^s}{\partial\phi^i_{(\lambda_1\dots\lambda_{k-1})}}
\equiv\delta^{(\mu)}_{(\nu(\lambda))}\frac{\partial}{\partial\phi^i_{(\lambda)}},
\end{equation}
which follows directly from the various definitions. In
multi-index notation, it becomes
\begin{equation}
  \left[\frac{\partial}{\partial\phi^i_{(\mu)}},\partial_\nu \right]
  =\delta^{(\mu)}_{((\lambda)\nu)}\frac{\partial}{\partial\phi^i_{(\lambda)}}.
\end{equation}
With this, we are ready to show \eqref{condsuff}:
\begin{equation}
\begin{split}
  \frac{\delta}{\delta\phi^i}(\partial_\mu j^\mu) &
  =(-)^{|\nu|}\partial_{(\nu)}
  \left(\frac{\partial}{\partial\phi^i_{(\nu)}}(\partial_\mu j^\mu)\right) \\
  & =
  (-)^{|\nu|}\partial_{(\nu)}\left(\left[\frac{\partial}{\partial\phi^i_{(\nu)}},\partial_\mu
    \right]j^\mu\right)
  +(-)^{|\nu|}\partial_{(\nu)}\partial_\mu\left(\frac{\partial j^\mu}{\partial\phi^i_{(\nu)}}\right) \\
  &
  =(-)^{|\nu|}\delta^{(\nu)}_{(\mu(\lambda))}\partial_{(\nu)}\frac{\partial
    j^\mu}{\partial\phi^i_{(\lambda)}}
  +(-)^{|\nu|}\partial_{(\nu)}\partial_\mu\frac{\partial j^\mu}{\partial\phi^i_{(\nu)}} \\
  & = (-)^{|\lambda|+1}\partial_{(\mu(\lambda))}\frac{\partial
    j^\mu}{\partial\phi^i_{(\lambda)}}+(-)^{|\lambda|}\partial_{(\lambda)\mu}\frac{\partial
    j^\mu}{\partial\phi^i_{(\lambda)}}=0.
\end{split}
\end{equation}

We now turn to the proof of the other implication, namely
\begin{align}
f=\partial_\mu j^\mu && \Longleftarrow && \frac{\delta f}{\delta\phi^i}=0.
\end{align}
We can write
\begin{equation}
\begin{split}
f[x,\phi]-f[x,0] & = \int_0^1d\lambda\frac{d}{d\lambda}f(x,\lambda\phi)\\
&
=\int_0^1\frac{d\lambda}{\lambda}(\phi^i_{(\nu)}\frac{\partial}{\partial\phi^i_{(\nu)}}f)[x,\lambda\phi)].
\label{hom}
\end{split}
\end{equation}
By repeated ``integrations by parts''\footnote{There is no actual
  integration involved, what we mean is the use of Leibniz' rule as
  appropriate in the context of integrations by parts.}, we have on
the one hand that
\begin{equation}
  \label{eq:1}
  \phi^i_{(\nu)}\frac{\partial}{\partial\phi^i_{(\nu)}}f=\phi^i\frac{\delta
    f}{\delta\phi^i}+\partial_\mu \tilde k^\mu[x,\phi], 
\end{equation}
with $\tilde k^\mu[x,0]=0$. 
On the other hand, $f[x,0]=\frac{\partial}{\partial x^\mu}\xi^\mu(x)$ for some
$\xi^\mu(x)$ on account of the ordinary Poincar\'e lemma for
$\mathbb{R}^n$ (of which we will provide a proof at the end of this
lecture, as it is the most basic and prototypical instance of
a cohomological computation).
When using these relations in \eqref{hom}, together with the
assumption $\frac{\delta f}{\delta\phi^i}=0 $, we find
\begin{equation}
  \label{eq:2}
  f=\partial_\mu j^\mu,\quad
  j^\mu=\xi^\mu+\int_0^1\frac{d\lambda}{\lambda}\tilde k^\mu[x,\lambda\phi]. 
\end{equation}
\end{proof}

\begin{center}
  ----------
\end{center}

Note also that, according to \eqref{hom}, local functions can be
decomposed as $f[x,\phi]=f[x,0]+\tilde f[x,\phi]$, with
$\tilde f[x,0]=0$, and that Lemma \ref{Lemma1} holds in the space of
local functions $\tilde f$ with no field-independent part.

\section{Local functionals}

A \textbf{local functional}
is the integral, on the base manifold of the jet bundle, of a local
function evaluated at sections $\phi^i(x)$ (of compact support or
vanishing at the boundary),
\begin{equation}
F[\tilde f,\phi(x)]=\int_Md^nx\,\tilde{f}|_s.
\end{equation}
\begin{lemma}
For two local functionals $F,G$, we have
\begin{equation}
F[\tilde f,\phi(x)]=G[\tilde g,\phi(x)]
\end{equation}
$\forall \phi^i(x)$ if and only if
\begin{equation}
\tilde f=\tilde g+\partial_\mu \tilde j^\mu,
\end{equation}
which in turn is equivalent to 
\begin{equation}
\frac{\delta}{\delta\phi^i}(\tilde f-\tilde g)=0,
\end{equation}
because of
Lemma \ref{Lemma1}.
\end{lemma}

\begin{center}
  ----------
\end{center}

\subsection{Exercise 2}
\label{sec:exercise-2}

Prove Lemma 2.

\begin{proof}
The implication 
\begin{align}
\tilde f=\tilde g+\partial_\mu \tilde j^\mu && \Longrightarrow &&
F[\tilde f,\phi(x)]=G[\tilde g,\phi(x)]
\end{align}
is almost obvious. In fact, we have
\begin{equation}
\begin{split}
  F[\tilde f,\phi(x)] & =\int d^nx\, \tilde f|_s=\int d^nx\, \tilde
  g|_s+\int
  d^nx\,(\partial_\mu \tilde j^\mu)|_s \\
  & = G[\tilde g,\phi(x)]+\int d^nx\,\frac{d}{dx^\mu}(j^\mu
  |_s)=G[\tilde g,\phi(x)],
\end{split}
\end{equation}
where we have used property \eqref{eq:1totalder} of total derivatives
evaluated at a section, together with the fact that we are assuming
suitable fall-off or boundary conditions on the fields.

Proving the other implication is also straightforward: if
$F[\tilde f,\phi(x)]=G[\tilde g,\phi(x)]$, we must have
\begin{equation}
\int d^nx\, (\tilde f-\tilde g)|_s=0.
\end{equation}
Since this must be true for all $\phi^i(x)$, the functional
$\int d^nx(\tilde f-\tilde g)|_s $ is a constant in $\phi(x)$. This means that its
functional derivative must vanish:
\begin{equation}
\frac{\delta}{\delta\phi^i(x)}\int d^n x(\tilde f-\tilde g)=0.
\end{equation}
But the vanishing of the functional derivative of a local functional
means the vanishing of the Euler-Lagrange derivative of its associated
local function when all boundary terms can be neglected, i.e.,
\begin{equation}
\frac{\delta \tilde f}{\delta\phi^i}=\frac{\delta \tilde g}{\delta\phi^i},
\end{equation}
which is what we wanted to prove.
\end{proof}

\begin{center}
  ----------
\end{center}

This
motivates the following algebraic definition of local functionals
\cite{Gelfand:1979zh}: a local functional is an equivalence class of
local functions
\begin{equation}
F[\tilde f,\phi(x)]\longleftrightarrow \left\{[\tilde f]:\tilde f\sim \tilde
  g\Leftrightarrow\frac{\delta}{\delta\phi^i}(\tilde f-\tilde g)=0 \right\},
\end{equation}
so that we can stop thinking about functionals and functional
analysis, and concentrate on suitable equivalence classes of local
functions instead.

\section{The Horizontal Complex}

Up to now we have mainly discussed vector fields on the jet bundle. As
in the case of ordinary differential geometry, we will now introduce
suitable differential forms. A \textbf{horizontal differential
  form}
on the jet bundle is an object
\begin{equation}
\omega=\sum_{l=0}^n\sum_{0\le\mu_1<\dots<\mu_l\leq n-1}\omega_{\mu_1\dots\mu_l}dx^{\mu_1}\dots dx^{\mu_l},
\end{equation}
where $\omega_{\mu_1\dots\mu_l} $ are local functions on the jet
bundle. We will omit the wedge products between forms and treat the
$dx$'s as anticommuting Grassmann odd variables instead. In this local
context, forms in top form degree $n$ correspond to Lagrangian
densities times the volume form. The \textbf{horizontal differential}
$d_H$ is then defined as
\begin{equation}
\begin{split}
d_H: 	& \Omega^p\longrightarrow\Omega^{p+1} \\
  		& d_H\omega^p=dx^\mu\partial_\mu\omega^p.
\end{split}
\end{equation}

As a consequence of the fact that total derivatives commute, the horizontal differential is nilpotent:
\begin{align}
[\partial_\mu,\partial_\nu]=0 &&\implies && d_H^2=0.
\end{align}
This allows us to define the cohomology of $d_H$: we define
p-\textbf{cocycles} $Z^p$ and p-\textbf{coboundaries} $B^p$ as
\begin{align}
\omega^p\in Z^p && \Longleftrightarrow && d_H\omega^p=0, \\
\omega^p\in B^p && \Longleftrightarrow && \omega^p=d_H\eta^{p-1},
\end{align}
and as usual we define the cohomology groups $H^p$ as
\begin{equation}
H^p(d_H,\Omega)\simeq Z^p/B^p.
\end{equation}
The cohomology of differential forms on $\mathbb{R}^n$ is encoded in
the Poincaré lemma, which implies that any closed form with form
degree at least 1 is exact. This theorem is generalized to the
present case, that of a trivial jet bundle, but before discussing this
generalization, let us review the ordinary Poincaré lemma for the de
Rham complex in $\mathbb{R}^n $.
\begin{theorem}(Poincaré lemma)
  The only non-trivial de Rham cohomology group of $\mathbb{R}^n $ is
  in form degree $0$:
\begin{equation}
H^p(d,\Omega)=\delta^p_0\mathbb{R}.
\end{equation}
\begin{proof}
  Let us define the operator
\begin{equation}
\rho=x^\mu\frac{\partial}{\partial dx^\mu}.
\end{equation}
Note that, if we define the vector field $v=x^\mu\frac{\partial}{\partial x^\mu}$, we have
\begin{equation}
\rho(dx^\mu)=x^\nu\frac{\partial}{\partial dx^\nu}dx^\mu=x^\mu=\iota_v(dx^\mu),
\end{equation}
so that we are in fact representing the operation of contracting
differential forms along vector fields by a suitable operator
involving the Grassmann odd variables $dx^\mu$. In this context, we
have 
\begin{equation}
  N\equiv\{d,\rho \}=\left(x^\mu\frac{\partial}{\partial x^\mu}
    +dx^\mu\frac{\partial}{\partial dx^\mu} \right).
\end{equation}
Using this, we can write for a generic $p$-form
$\omega\in\Omega^p(\mathbb{R}^n)$,
\begin{equation}\label{eq:Poinproof1}
\begin{split}
  \omega(x,dx)-\omega(0,0) & =
  \int_0^1d\lambda\frac{d}{d\lambda}\omega(\lambda x,\lambda dx) \\
 &
 =\int\frac{d\lambda}{\lambda}\left[\left(x^\mu\frac{\partial}{\partial
       x^\mu}
     +dx^\mu\frac{\partial}{\partial dx^\mu} \right)\omega \right](\lambda x,\lambda dx) \\
 & =\int_0^1\frac{d\lambda}{\lambda}\left(\{d,\rho\}\omega \right)(\lambda x,\lambda dx)
\end{split}
\end{equation}
In particular, if $d\omega=0$, equation \eqref{eq:Poinproof1} becomes
\begin{equation}
\omega(x,dx)=\omega(0,0)+d I(\omega), 
\end{equation}
where 
\begin{equation}
  \label{eq:3}
  I(\omega)=\int_0^1\frac{d\lambda}{\lambda}(\rho\omega)(\lambda
x,\lambda dx),
\end{equation}
sends $p$-forms to $p-1$ forms.
We now see that all closed forms in $\mathbb{R}^n$ are exact, except
for the constant 0-forms,
\begin{equation}
H^p(d,\Omega)=\delta^p_0\mathbb{R},
\end{equation}
which is what we wanted to prove.
\end{proof}
\end{theorem}
\begin{remark} 
  As a consequence of the Poincaré lemma, every top form, which is
  necessarily closed, must be exact. If Lagrangian densities would
  correspond to top forms in the de Rham complex, they all would be
  given by a divergence. Since there are Euler-Lagrange equations
  giving rise to non-trivial dynamics, the Poincaré lemma must be
  modified when we go from the de Rham to the horizontal complex.
\end{remark}

\begin{center}
  ----------
\end{center}

\subsection{Exercise 3}
\label{sec:exercise-3}

Show that $f(x)=\frac{\partial}{\partial x^\mu}\xi^\mu(x)$, with
$ \xi^\mu=\int^1_0d\lambda \lambda^{n-1} x^\mu f(\lambda x)$ by using
the Poincar\'e lemma in top form degree $n$. If $g_\mu(x)$ are $n$
functions that satisfy the integrability conditions
$\frac{\partial}{\partial x^\mu} g_\nu-\frac{\partial}{\partial x^\nu}
g_\mu=0$, show that $g_\mu=\frac{\partial}{\partial x^\mu} f(x)$ with
$f=x^\mu\int^1_0d\lambda g_\mu(\lambda x)$ by using the Poincar\'e
lemma in form degree $1$.

\begin{center}
  ----------
\end{center}

We will now state a generalization of the Poincaré lemma to
the field theoretical case, giving an idea of how the proof goes in
analogy to the ordinary Poincaré lemma.

\begin{theorem}\label{th:1cohomo}(Algebraic Poincaré lemma)
  The cohomology groups of the horizontal complex are
\begin{equation}
H^p(d_H,\Omega(E))=
\begin{cases}
\mathbb{R}, & p=0,\\
0, & 0<p<n, \\
[\omega^n], & p=n
\end{cases}
\end{equation}
where we have denoted by $[\omega^n]$, suitable equivalence classes of top
forms,
\begin{multline} [\omega^n]=\Big\{\omega\in\Omega^n(E):
    \omega\sim\omega'\Leftrightarrow\omega=\omega'+d_Hj^{n-1}\\
    \Leftrightarrow\frac{\delta}{\delta\phi^i}(\omega-\omega')=0
  \Big\}.
\end{multline}
\begin{proof}
  The idea of the proof is the same as for the ordinary Poincaré
  lemma. The analog of the number operator in the previous proof is
  the operator
\begin{equation}
\delta_Q\equiv\partial_{(\mu)}Q^i\frac{\partial}{\partial\phi^i_{(\mu)}},
\end{equation}
which commutes with the total derivative and implements an
\textbf{infinitesimal field variation}. Such a variation can be
written in a unique way as 
\begin{equation}
\delta_Qf=\sum_{|\mu|\le k}\partial_{(\mu)}\left[Q^i\frac{\delta
    f}{\delta\phi^i_{(\mu)}} \right], 
\end{equation}
where
\begin{equation}
\frac{\delta f}{\delta\phi^i_{(\mu)}}=\sum_{|\nu|\le
  k-|\mu|}\binom{|\mu|+|\nu|}{|\mu|}(-)^{|\nu|}
\partial_{(\nu)}\frac{\partial f}{\partial\phi^i_{\left((\mu)(\nu) \right)}}
\end{equation}
are higher order Euler-Lagrange derivatives. One then defines the
operator
\begin{equation}
I^p_Q(\omega^p)=\sum_{|\lambda|\le
  k-1}\frac{|\lambda|+1}{n-p+|\lambda|+1}\partial_{(\lambda)}
\left[Q^i\frac{\delta}{\delta\phi^i_{((\lambda)\rho)}}
  \frac{\partial \omega^p}{\partial dx^\rho} \right]
\end{equation}
that sends $p$-forms to $p-1$-forms. The difficult part of the proof,
which we omit here,
is to show that 
\begin{equation}
\begin{cases}
\delta_Q\omega^p=I^{p+1}_Q(d_H\omega^p)+d_H(I^p_Q\omega^p), & 0\le p < n, \\
\delta_Q\omega^n=Q^i\frac{\delta\omega^n}{\delta\phi^i}+d_H(I^n_Q\omega^n).
\end{cases}
\end{equation}
The proof proceeds then in the same way as in the ordinary Poincaré
lemma by using 
\begin{align}
\rho_H\tilde\omega^p=\int_0^\lambda\frac{d\lambda}{\lambda}(I_{H,\phi}\tilde\omega^p)[x,\lambda\phi],
  && \{d_H,\rho_H \}\tilde\omega^p=\tilde \omega^p,
\end{align}
for $p<n$. 
\end{proof}
\end{theorem}

\chapter{Dynamics and the Koszul-Tate complex}

After having laid down the basic concepts, in this second lecture we
start dealing with dynamics. Noether identities and the associated
Koszul-Tate differential \cite{Fisch:1989dq,Fisch:1990rp} are
introduced in the framework, allowing us to deal with gauge theories.

More details can be found in \cite{Henneaux:1992ig} and also
in \cite{Barnich:2000zw,Barnich:2001jy}.

\section{Stationary surface}

In field theory we usually deal with partial differential
equations. In particular, we select those that arise from an action
principle because we know how to quantize them. The action is a local
functional on the jet bundle
\begin{equation}
S[L,\phi(x)]=\int_Md^nx\, L[x,\phi]|_{\phi(x)},
\end{equation}
and the equations of motion 
derived from the action principle are 
\begin{equation}\label{eq:2EL}
\frac{\delta L}{\delta\phi^i}=0.
\end{equation}
In order to obtain solutions $\bar{\phi}^i(x)$, we have to evaluate
the left hand side of \eqref{eq:2EL} at sections and find those that
solve the associated partial differential equations. For our purpose
here, it is enough to exchange this task with that of studying some
algebraic/geometric properties of hypersurfaces in the jet bundle. In
fact, for most cases of physical interest, when the Euler-Lagrange
equations do not contain derivatives of order higher than one,
equation \eqref{eq:2EL} defines a surface in the second jet space
$J^2(E)$. However, even in these cases, we need to take into account
higher derivatives of the equations of motion. More generally then, we
consider the \textbf{stationary surface}
defined in the jet-bundles $J^k(E)$ by the Euler-Lagrange equations
and an appropriate number of their total derivatives,
\begin{equation}
  \label{eq:4}
  \partial_{(\mu)}\frac{\delta L}{\delta\phi^i}=0. 
\end{equation}

A local function $f$ is \textbf{weakly zero}
if it vanishes when one pulls it back to the stationary surface. In
this case, one uses Dirac's notation and writes
\begin{equation}
f\approx 0.
\end{equation}
In fact, because of standard regularity assumptions on the functions
that define this surface, if an object vanishes weakly, then it must
be a linear combination of the equations defining the surface. This
means that
\begin{align}
  f\approx 0 && \Longleftrightarrow &&
 f=k^{i(\mu)}\partial_{(\mu)}\frac{\delta L}{\delta\phi^i}.
\end{align}

Let us work out an example in order to clarify the content of this
construction in a simple setting:
\begin{example}(Massive scalar field) Let us consider a free real
  scalar field $\phi$ with mass $m$. In this case the Lagrangian is
\begin{equation}
  L^{KG}=-(\frac{1}{2}\partial_\mu\phi\partial^\mu\phi+\frac{m^2}{2}\phi^2).
  \label{kleingordon}
\end{equation}
The equation of motion for this Lagrangian is the massive Klein-Gordon
equation
\begin{equation}
  -\frac{\delta L^{KG}}{\delta\phi}=-\partial^\mu\partial_\mu\phi+m^2\phi
  =\phi_{00}-\phi_{ii}+m^2\phi=0
\end{equation}
In our context, this field equation is an algebraic equation between
the local coordinates of the jet space $\phi$, $\phi_\mu$,
$\phi_{\mu_1\mu_2}$. 
\end{example}

\begin{center}
  ----------
\end{center}

\subsection{Exercise 4}
\label{sec:exercise-4}

Compute the Euler-Lagrange derivatives of

(i) Yang-Mills theory
\begin{equation}
  \label{eq:5}
  L^{YM}[A_\mu^\alpha]=-\frac{1}{4 g^2}
  F^\alpha_{\mu\nu} F^{\beta\mu\nu} g_{\alpha\beta},\quad
    F^\alpha_{\mu\nu}=\partial_\mu A^\alpha_\nu-\partial_\nu
    A_\mu^\alpha+f^\alpha_{\beta\gamma}A^\beta_\mu A^\gamma_\nu,
\end{equation}
where $g_{\alpha\beta}$ is an invariant symmetric tensor on a Lie algebra
$\mathfrak g$ with structure constants $f^\alpha_{\beta\gamma}$,

(ii) Chern-Simons theory in 3 spacetime dimensions, 
\begin{equation}
  \label{eq:6}
  L^{CS}[A_\mu^\alpha]=\frac{k}{8\pi}g_{\alpha\beta}\epsilon^{\mu\nu\rho} A_\mu^\alpha(\partial_\nu
  A^\beta_\rho+\frac{1}{3}f^\beta_{\gamma\delta}A_\nu^\gamma A_\rho^\delta), 
\end{equation}

(iii) Einstein gravity
\begin{equation}
  \label{eq:7}
  L^{EH}[g_{\mu\nu}]=\frac{1}{16\pi G}\sqrt{|g|} (R-2\Lambda). 
\end{equation}
Hint: Use that 
\begin{equation*}
\delta L=\delta\phi^i\frac{\delta
  L}{\delta\phi^i}+\partial_\mu(\cdot)^\mu
\end{equation*}
defines the Euler-Lagrange
equations uniquely, that 
\begin{equation*}
\delta
{\Gamma^\mu}_{\nu\lambda}=\frac{1}{2} g^{\mu\rho}(\delta
g_{\rho\nu;\lambda}+\delta g_{\rho\lambda;\nu}-\delta
g_{\nu\lambda;\rho}),
\end{equation*}
and that $\sqrt{|g|}
{j^\mu}_{;\mu}=(\sqrt{|g|} j^\mu)_{,\mu}$ implies
\begin{multline*}
  \sqrt{|g|}l^{\nu_1\dots\nu_n\mu}m_{\nu_1\dots\nu_n;\mu}
  =
  (\sqrt{|g|}l^{\nu_1\dots\nu_n\mu}m_{\nu_1\dots\nu_n})_{,\mu}\\
  -
  \sqrt{|g|}{l^{\nu_1\dots\nu_n\mu}}_{;\mu}m_{\nu_1\dots\nu_n}. 
\end{multline*}

(iv) the Poisson-Sigma model \cite{Ikeda:1993fh,Schaller:1994es} in 2
spacetime dimensions,
\begin{equation}
  \label{eq:10}
  L^{PSM}[\eta_{i\mu},X^i]d^2x =\eta_idX^i +\frac{1}{2}
  \alpha^{ij}(X)\eta_i\eta_j, 
\end{equation}
where $\eta_i=\eta_{i\mu}dx^\mu$ and $\alpha^{ij}(X)$ defines a
Poisson tensor on target space,
\begin{equation*}
  \label{eq:11}
  \alpha^{ij}=-\alpha^{ji},\quad \alpha^{il}\frac{\partial
    \alpha^{jk}}{\partial X^l}+{\rm cyclic} (i,j,k)=0,
\end{equation*}
so that the Poisson bracket 
\begin{equation*}
  \label{eq:12}
  \{f(X),g(X)\}=\frac{\partial f}{\partial
    X^i}\alpha^{ij}\frac{\partial g}{\partial X^j}
\end{equation*}
is skew symmetric, satisfies the Leibniz rule and the Jacobi identity. 

\section{Noether identities}
\label{sec:noether-identities}

\textbf{Noether identities}
are relations between the equations determining the stationary
surface:
\begin{equation}
  N^{i(\mu)}\partial_{(\mu)}\frac{\delta L}{\delta\phi^i}
  \equiv N^i\left[\frac{\delta L}{\delta\phi^i} \right]=
  0.\label{eq:211} 
\end{equation}
Here $N^{i(\mu)}$ are local functions. Let us consider an illustrative
example, which may appear a little unusual:
\begin{example}
  Consider several scalar fields satisfying the Klein-Gordon
  equation. In this case, we have a set of Noether identities given by
\begin{equation}
N^i=\frac{\delta L}{\delta\phi^j}\mu^{[ji]},
\end{equation}
where $\mu^{[ji]}$ are a collection of local functions that are
antisymmetric in their indices:  
\begin{equation}
\mu^{[ij]}\frac{\delta L}{\delta\phi^i}\frac{\delta L}{\delta\phi^j}=0. 
\end{equation} 
Note that the presence of Noether identities is not due to the fact
that we are considering multiple scalar fields: in the case of a
single field, in fact, we have for example a Noether identity
determined by
\begin{equation}
N=\partial_\nu\left(\frac{\delta L}{\delta\phi} \right)M^{[\nu\mu]}\partial_\mu.
\end{equation}
\end{example}
The ones we have just shown are examples of \textbf{trivial Noether
  identities}: a trivial Noether identity is a Noether identity with
coefficients vanishing on-shell, i.e. characterized by
\begin{equation}
N^{i(\mu)}\partial_{(\mu)}\approx 0.
\end{equation}

A \textbf{proper gauge theory}
is defined as a theory for which there are nontrivial Noether
identities. For example, in electromagnetism,
\begin{equation}
L^{Max}(A_\mu)=-\frac{1}{4}F_{\mu\nu}F^{\mu\nu},\quad \frac{\delta L^{Max}}{\delta
  A_\mu}=\partial_\nu F^{\nu\mu}, \label{maxwell}
\end{equation}
we have
\begin{equation}
\partial_\mu\partial_\nu F^{\mu\nu}=0.
\end{equation}
with corresponding nontrivial Noether identity
\begin{equation}
 N^{(\mu)\nu}\partial_{(\mu)}=\eta^{\mu\nu}\partial_\mu\ne0.
\end{equation}
Other examples are Yang-Mills and Chern-Simons theories for which the
Noether identities take the form
\begin{equation}
  D_\mu\frac{\delta L^{YM; CS}}{\delta A^\alpha_\mu}=0,
\end{equation}
where $D_\mu f^\alpha=\partial_\mu f^\alpha+A^\beta_\mu
f^\alpha_{\beta\gamma} f^\gamma$
(adjoint
representation), while $D_\mu g_\alpha=\partial_\mu g_\alpha-A^\beta_\mu f^\gamma_{\beta\alpha}
g_\gamma$ (co-adjoint representation). 

\begin{center}
  ----------
\end{center}

\subsection{Exercise 5}
\label{sec:exercise-5}

In the case of Einstein gravity, show that the Noether identities
\begin{equation}
  \label{eq:8}
  \partial_\nu \frac{\delta L^{EH}}{\delta
    g_{\mu\nu}}+\Gamma^\mu_{\nu\lambda}
  \frac{\delta L^{EH}}{\delta g_{\nu\lambda}}=0
  \end{equation}
  are equivalent to the contracted Bianchi identities
  \begin{equation}
    \label{eq:9}
    {G^{\mu\nu}}_{;\nu}=0. 
  \end{equation}

  \section{Irreducible gauge theories}
\label{sec:irred-gauge-theor}

  One typical problem in field theory is to find all the Noether
  identities of a given theory, a problem which can be naturally
  tackled in this framework. In particular, we say that we have an
  \textbf{irreducible gauge theory} if we have a set of Noether
  identities $\{R_\alpha^{\dagger i(\mu)}\partial_{(\mu)} \}$ such that 
\begin{equation}
  \label{Noether}
  R^{\dagger i}_\alpha\left[\frac{\delta L}{\delta\phi^i} \right]=0,
\end{equation}
which is
\begin{itemize}
\item \textbf{Non-trivial}:
\begin{equation}\label{eq:2nontrivial}
R_\alpha^{\dagger i(\mu)}\not\approx0;
\end{equation}
\item \textbf{Irreducible}: this means that there is no combination of
  the Noether identities and their derivatives which is trivial:
\begin{align}\label{eq:2irreduciba}
  Z^{\dagger\alpha(\nu)}\partial_{(\nu)}\circ
  R_\alpha^{\dagger i(\mu)}\partial_{(\mu)}
  \approx 0 && \implies && Z^{\dagger \alpha(\beta)}=0;
\end{align}
\item Constitutes a \textbf{generating set}: 
  this means that any Noether operator,

  \noindent $N^{i(\mu)}\partial_{(\mu)}$
  satisfying \eqref{eq:211}, may be written in the form
\begin{equation}
  N^{i(\lambda)}\partial_{(\lambda)}=Z^{\dagger \alpha(\nu)}\partial_{(\nu)}\circ 
  R^{\dagger i(\mu)}_\alpha\partial_{(\mu)}+M^{[j(\nu)i(\lambda)]}\partial_{(\nu)}
  \frac{\delta L}{\delta\phi^i}\partial_{(\lambda)},\label{eq:2irreducibb}
\end{equation}
i.e., as a sum of two terms, one given by the Noether identities of the
generating set and the other vanishing on-shell (or, more precisely,
an antisymmetric combination of the left hand sides of the equations
of motion and their derivatives). 
\end{itemize}

To find such a generating set is in general hard. In the examples
above, a slightly tedious analysis of the left hand sides of the
equations of motion of the theory and their derivatives in jet space
shows that the standard ones given above constitute such a generating
set. As we will see below, such an analysis is in fact equivalent to
showing that one has found all the non-trivial gauge symmetries of the
theory.

\section{The Koszul-Tate complex}

We have already introduced a nilpotent operator, the horizontal
differential $d_H$, together with its cohomological complex in order
to have algebraic control on functionals. Another differential, the
\textbf{Koszul-Tate differential} $\delta$, and its cohomology will
take care of the dynamics of the theory. In the case of irreducible
gauge theories, it is constructed as follows.

For each equation of motion $\frac{\delta L}{\delta\phi^i}$, one
introduces a Grassmann odd field $\phi^*_{i}$, whereas for each
non-trivial Noether operator of the generating set
$R^{+i(\mu)}_\alpha\partial_{(\mu)}$ (if any), one introduces a
Grassmann even field $C^*_{\alpha}$.  The definition of the total
derivative operator $\partial_\mu$, of local functions, the horizontal
complex and local functionals is then extended so as to include (a
polynomial dependence of) these additional fields and their
derivatives. In terms of generators, the Koszul-Tate differential
may then by defined through
$\delta\phi^i=0=\delta x^\mu=\delta dx^\mu$, while
\begin{equation}
\begin{cases}
\delta\phi_i^*=\frac{\delta L}{\delta\phi^i}, \\
\delta C_\alpha^*=R^{\dagger i}_\alpha[\phi_i^*],
\end{cases}
\end{equation}
with the understanding that 
\begin{equation}
[\delta,\partial_\mu]=0.\label{eq:13}
\end{equation}

Equivalently, it may be represented by the odd vector field
\begin{equation}
  \delta=\partial_{(\mu)}\frac{\delta L}{\delta\phi^i}
  \frac{\partial}{\partial\phi^*_{i(\mu)}}+\partial_{(\nu)}
  \left(R_\alpha^{\dagger i}[\phi_i^*] \right)\frac{\partial}{\partial C^*_{\alpha(\nu)}}.
\end{equation}
Further, assuming that all odd variables anticommute between
themselves (so that, for example, $dx^\mu$'s anticommute with
$\phi^*_i$'s), \eqref{eq:13} becomes
\begin{equation}
\{\delta,d_H \}=0.
\end{equation}
This vector field is nilpotent,
\begin{equation}
\delta^2=\frac{1}{2}\{\delta,\delta\}=0,
\end{equation}
on account of \eqref{Noether}. 
 
The introduction of the additional fields
provides an additional $\mathbb{Z}$-grading of the extended horizontal
complex, the \textbf{antifield number}, which can be described by the
vector field
\begin{equation}
  \mathcal{A}=\phi^*_{i(\mu)}\frac{\partial}{\partial\phi^*_{i(\mu)}}
  +2C^*_{\alpha(\beta)}\frac{\partial}{\partial C^*_{\alpha(\beta)}}.
\end{equation}
Applied to a field, this operator gives back the same field multiplied
by its antifield number: so, for example,
\begin{align}
\mathcal{A}\phi^i=0, && \mathcal{A}\phi_i^{*2}=2\phi_i^{*2}.
\end{align}
If $\Omega_k$ denote the horizontal forms with definite antifield
number $k$,
then
\begin{equation}
\delta:\Omega_k\longrightarrow\Omega_{k-1}.
\end{equation}
Since $\delta$ reduces rather than
increases the degree, one refers to it as boundary rather than a
co-boundary operator. Since it is nilpotent, one may define its
homology. This is the object of the following:
\begin{theorem}
\begin{equation}\label{eq:KTCoho}
H_k(\delta,\Omega)=
\begin{cases}
0, & k>0, \\
\Omega(\Sigma)\simeq\Omega_0/I, & k=0
\end{cases}
\end{equation}
where we have defined the set of forms that vanish after the pull-back
to $\Sigma$ as
\begin{equation}
I=\{\omega\in \Omega:\omega\approx 0 \}.
\end{equation}
\end{theorem}
\begin{remark}
  A convenient way of describing functions (or horizontal forms) on a
  surface embedded in a larger space is to consider the functions
  defined on the whole space modulo those that vanish on the
  surface. It is in this sense that one says that the construction
  above provides a homological resolution of the horizontal complex
  $\Omega(\Sigma)$ associated to the stationary surface.
\end{remark}
\begin{remark}
  Crucial in the proof of the first line of the theorem (also referred
  to as ``acyclicity of $\delta$ in higher antifield number'') are
  properties \eqref{eq:2irreduciba} and \eqref{eq:2irreducibb} of the
  generating set. For reducible gauge theories (for instance theories
  involving $p-$ forms with $p>1$), the construction has to be
  suitably modified in order for the theorem to continue to hold.
\end{remark}

\section{Characteristic cohomology }
\label{sec:char-cohom-}

Characteristic cohomology
$H^p(d_H,\Omega(\Sigma))$ contains important physical
information related to the equations of motion.

In the extended
complex, it is encoded through
\begin{equation}\label{eq:2equivHom}
H^p(d_H,\Omega(\Sigma))\simeq H_0^p(d_H|\delta).
\end{equation}
Indeed, consider first a weakly closed $p$-form with zero antifield
number
\begin{equation}
d_H\omega_0^p\approx 0.
\end{equation}
Since any (weakly) vanishing function is a linear combination of the
equations of motion and their total derivatives, we can write
\begin{equation}
d_H\omega_0^p=k^{i(\mu)} \partial_{(\mu)}\frac{\delta L}{\delta\phi^i}.
\end{equation}
This in turn is equivalent to
\begin{equation}
d_H\omega_0^p+\delta\omega_1^{p+1}=0,
\end{equation}
where
\begin{equation}
	\omega_1=-k^{i(\mu)}\phi_{i(\mu)}^*.
\end{equation}
For the coboundary condition, every weakly exact form
\begin{equation}
\omega_0^p\approx d_H\eta_0^{p-1},
\end{equation}
for the same reason as above, can be written as
\begin{equation}
  \omega_0^p=d_H\eta_0^{p-1}+\delta\eta_1^p.\label{trivial}
\end{equation}

$H^{n-1}(d_H,\Omega(\Sigma))$, the characteristic cohomology in degree
$n-1$, classifies the non-trivial conserved currents
of the theory. Indeed, an $n-1$-form of zero antifield number can be
written as
\begin{equation}
\omega_0^{n-1}=j^\mu\left(d^{n-1}x \right)_\mu, 
\end{equation}
with
\begin{equation}
  \label{eq:15}
  \left(d^{n-k}x\right)_{\mu_1\dots\mu_k}
 =\frac{1}{k!(n-k)!}\epsilon_{\mu_1\dots\mu_{k}\nu_{k+1}\dots\nu_n}dx^{\nu_{k+2}}\dots
 dx^{\nu_n}. 
\end{equation}
This implies that
\begin{align}
d_H\omega_0^{n-1}\approx 0, && \Leftrightarrow && \partial_\mu j^\mu\approx 0.
\end{align}
Similarly,
\begin{equation}
  \eta_0^{n-2}=k^{[\mu\nu]}\left(d^{n-2}x\right)_{\mu\nu}, \label{eq:14}
\end{equation}
so that a trivial conserved current as in \eqref{trivial} is
explicitly given by 
\begin{equation}
j^\mu=\partial_\nu k^{[\mu\nu]}+t^\mu,
\end{equation}
with $t^\mu\approx 0$. In other words, trivial conserved currents are
given by divergences of ``superpotentials'' $k^{[\mu\nu]} $ and by
on-shell vanishing ones.

\begin{example}(Energy-momentum tensor)
For a scalar field with Lagrangian given in \eqref{kleingordon},
the energy-momentum tensor 
\begin{equation}\label{eq:2scalarEMT}
T{^\mu}_\nu=\partial^\mu\phi\partial_\nu\phi-\delta{^\mu}_\nu L^{KG},
\end{equation}
provides 4 non-trivial conserved currents since one may easily check
that

\noindent $\partial_\mu T{^\mu}_\nu\approx 0$. The same holds for
electromagnetism with Lagrangian given in \eqref{maxwell} for which
\begin{equation}\label{eq:2EMEMT}
T{^\mu}_\nu=F^{\mu\sigma}\partial_\nu A_\sigma-\delta{^\mu}{_\nu}L^{Max}.
\end{equation}
\end{example}
\begin{remark}
  A concrete expression for a non-trivial conserved current is merely
  a representative of an equivalence class defined up to the addition
  of trivial conserved currents. This is important for instance when
  one wants to couple matter fields to gravity in the context of
  general relativity, where a symmetric energy momentum-tensor is
  required so that it can be contracted with the metric. This is the
  case for the energy-momentum tensor $T^{\mu\nu}$ of the scalar field
  but not for the one of electromagnetism given in
  \eqref{eq:2EMEMT}. In this context, the \textbf{Belinfante
    procedure} consists in constructing a symmetric representative of
  the energy-momentum tensor by using trivial conserved currents. 
\end{remark}

\begin{center}
  ----------
\end{center}

\subsection{Exercise 6}
\label{sec:exercise-6}
Show that the energy-momentum tensor \eqref{eq:2EMEMT} differs from
the symmetric expression 
\begin{equation}
T{^\mu}_\nu=-\left(F^{\mu\sigma}F_{\sigma\nu}+\frac{1}{4}F^{\rho\sigma}F_{\rho\sigma}\delta{^\mu}_\nu \right),
\end{equation}
by trivial terms.

\begin{center}
  ----------
\end{center}

For theories
with no non-trivial Noether identities, like scalar field theory for
example, there is no additional characteristic cohomology besides the
conserved currents in form degree $n-1$,
\begin{align}
H^p\left(d_H,\Omega(\Sigma) \right)=0, && p<n-1.
\end{align}

For proper irreducible gauge theories, it can be shown that there is
no characteristic cohomology in form degree $p<n-2$. The only
additional characteristic cohomology is thus in form degree
$p=n-2$. The conditions that a form be closed but not exact become in
this case
\begin{align}
\partial_\nu k^{[\mu\nu]}\approx 0, &&
 k^{\mu\nu}\not\approx\partial_\sigma\eta^{[\mu\nu\sigma]}. 
\end{align}
In electromagnetism, we will show that there is a single class with
representative given by $k^{\mu\nu}= F^{\mu\nu}$, which indeed satisfies
\begin{equation}
\partial_\nu F^{\mu\nu}\approx 0.
\end{equation}
This cohomology class captures electric charge contained in a closed
surface $S$ through  
\begin{equation}
Q=\oint_{t=cte,S} F^{\mu\nu}(d^{2}x)_{\mu\nu}=\int_{t=cte,S} F^{0i}d\sigma_i.
\end{equation}
In linearized gravity, ADM charges fall into this class.

The proofs of these statements, as well as of the generalized Noether
theorems to be discussed in the next section, rely on the following
isomorphisms between cohomology classes:
\begin{theorem}[inversion \& descent equations]\label{th4}
  \begin{equation}
    \label{eq:16}
    H_0^{n-p}(d_H|\delta)/\delta^n_p \mathbb R\simeq
    H_1^{n-p+1}(\delta|d_H)\simeq \dots\simeq H^n_p(\delta|d_H), 
  \end{equation}
\end{theorem}
\noindent which in turn can be proved by elementary homological techniques of
``diagram chasing''.

\chapter{Symmetries and longitudinal differential}

Up to now we have discussed Noether identities and conservation laws
from the perspective of the geometry and topology of the
surface of the equations of motion. For a stationary surface
originating from a variational principle, we will now relate the
former to gauge and the latter to global symmetries in terms of
complete and generalized Noether theorems. 

\section{Prolongation of transformations}

Given a section
\begin{equation}
  (x^\mu,\phi^i(x))
\end{equation}
of the jet bundle, a generalized vector field on $J^0(E)$ 
\begin{equation}
v=a^\mu\frac{\partial}{\partial x^\mu}+b^i\frac{\partial}{\partial\phi^i},
\end{equation}
where $a^\mu$ and $b^i$ are local functions, 
defines an infinitesimal transformation of a section, as
\begin{equation}\label{eq:3inftransf}
\left(x^\mu+\e a^\mu|_s,\phi^i(x)+\e b^i|_s \right).
\end{equation}
The first term thus encodes the action of the transformation on the
spacetime coordinates, while the second term encodes the variation of
the field. The transformation then produces the new section
\begin{equation}
\phi'^i(x+\e a)=\phi^i(x)+\e b^i|_s,
\end{equation}
so that the variation of the field will have not only the intrinsic
term $b^i$, but also a drag term including its derivative
\begin{align}
\delta_Q\phi^i\equiv Q^i, && Q^i=\left(b^i-a^\nu\phi^i_\nu \right).
\end{align}
$Q^i$ is called the \textbf{characteristic representative} of the
transformation. Through the characteristic representative we can
define the \textbf{prolongation} of the transformation from the fields
to their derivatives, i.e. to the infinite jet space, in terms of a
so-called {\bf evolutionary vector field},
\begin{equation}
\delta_Q=\partial_{(\mu)}Q^i\frac{\partial}{\partial\phi^i_{(\mu)}}.
\end{equation}
The prolongation is defined in such a way that the transformation
commutes with total derivatives:
\begin{equation}\label{eq:3variationder}
[\delta_Q,\partial_\mu]=0.
\end{equation}

\section{Symmetries of the equations of motion}
\label{sec:symm-equat-moti}

When studying field theories,
we can encounter two types of symmetries: symmetries of the
equations of motion and symmetries of the action. A symmetry of the
equations of motion is defined by
\begin{align}
\delta_Q\frac{\delta L}{\delta\phi^i}\approx 0.
\end{align}
In particular, this is equivalent to saying that given a section
$\bar{\phi}^i$ solving the Euler-Lagrange equations, then the
transformed section
\begin{equation}
\bar{\phi}^i(x)+\epsilon Q^i|_{\phi^i(x)}
\end{equation}
is a solution of the Euler-Lagrange equations to order $\epsilon^2$.

\section{Variational symmetries}
\label{sec:symmetries-action}

Symmetries of the action,
also called \textbf{variational symmetries} in this algebraic context,
are defined by
\begin{equation}
\delta_Q L
= \partial_\mu k^\mu_Q,
\end{equation}
for some local functions $k^\mu_Q$.  Notice that since we have defined
Lagrangians as equivalence classes up to total derivatives this
definition does not depend on the representative. We can rearrange the
above condition, using
\begin{equation}
\delta_QL=\partial_{(\mu)}Q^i\frac{\partial L}{\partial\phi^i_{(\mu)}}
\end{equation}
together with the definition of Euler-Lagrange derivative
\eqref{eq:1ELDer} in order to write
\begin{equation}
  Q^i\frac{\delta L}{\delta\phi^i}
  =\partial_\mu\left(k^\mu_Q-\frac{\partial L}{\partial\phi^i_\mu}Q^i
    +\dots \right)\equiv\partial_\mu j^\mu_Q,{}
\end{equation}
which is the usual statement of Noether's first theorem (to wit, that
a variational symmetry $Q^i$ gives rise to a conserved current
$j^\mu_Q$; we will return to that later).

A question one is led to ask is whether every variational symmetry 
is also a symmetry of the equations of motion and vice versa. In fact,
only the direct implication is true. Indeed, for general local
functions $f,g,Q^i$, one can show that
\begin{equation}
  \label{eq:17}
  (-)^{|\mu|}\partial_{(\mu)}\left[\frac{\partial
      (\partial_\nu f)}{\partial\phi^i_{(\mu)}} g
     \right]=-(-)^{|\mu|}\partial_{(\mu)}\left[\frac{\partial
       f}{\partial\phi^i_{(\mu)}} \partial_\nu g
     \right],
\end{equation}
\begin{equation}
  \delta_Q\frac{\delta f}{\delta\phi^i}=\frac{\delta}{\delta\phi^i}
  (\delta_Q f)-(-)^{|\mu|}\partial_{(\mu)}\left[\frac{\partial
      Q^j}{\partial\phi^i_{(\mu)}}
    \frac{\delta f}{\delta\phi^j} \right].\label{314}
\end{equation}
Applying the latter equation to the case $f=L$, with variational
$Q^i$, the first term on the right-hand side vanishes since the
Euler-Lagrange derivative of a total derivative is zero, 
\begin{equation}
\frac{\delta}{\delta\phi^i}(\partial_\mu k^\mu_Q)=0.
\end{equation}
The second term, on the other hand, is weakly zero since it is
proportional to the equations of motion and their derivatives. Thus we
find
\begin{equation}
\delta_Q\frac{\delta L}{\delta\phi^i}\approx 0
\end{equation}
i.e., that $Q^i$ is a symmetry of the equations of motion.

In order to see that the other implication does not hold in general,
it suffices to see a simple counterexample. For the massless
Klein-Gordon Lagrangian,
\begin{align}
L=-\frac{1}{2}(\partial_\mu\phi)^2,
\end{align}
the field equation is determined by 
\begin{equation}
\frac{\delta L}{\delta\phi}=\Box\phi.
\end{equation}
An equations of motion symmetry is given by 
\begin{equation}
Q=\lambda\phi,
\end{equation}
with $\lambda$ constant. However,
\begin{equation}
\delta_QL=2\lambda L,
\end{equation}
which cannot be a total derivative since the Euler-Lagrange equations
of motion are non-trivial, while the Euler-Lagrange derivative would
annihilate any total derivative.
\begin{remark}
Non-Noetherian symmetries, i.e.,
symmetries of variational equations which are not variational
symmetries, do play an important role in the context of integrable
systems.
\end{remark}
\begin{example}(Translations) For a translation in the direction of
  $x^\nu$, we have
\begin{align}
a^\mu_\nu=\delta^\mu_\nu, && b^i_\nu=0.
\end{align}
Equation \eqref{eq:3inftransf} then implies that 
\begin{equation}
Q^i_\nu=-\phi^i_\nu,
\end{equation}
with prolongation given by
\begin{equation}
  \delta_{Q_\nu}=-\phi^i_{((\mu)\nu)}\frac{\partial}{\partial\phi^i_{(\mu)}}
  =-\partial_\nu+\frac{\partial}{\partial x^\nu}.
\end{equation}
We see in particular that a Lagrangian with no explicit dependence on
spacetime coordinates is translation invariant,
\begin{align}
\frac{\partial L}{\partial x^\nu}=0 && \Longrightarrow &&
\delta_{Q_\nu}L=-\partial_\nu L=\partial_\mu( -L\delta^\mu_\nu),
\end{align}
with associated Noether current
$j^\mu_\nu=\frac{\partial L}{\phi^i_\mu}\phi^i_\nu-\delta^\mu_\nu
L\equiv {T^\mu}_\nu$, the (canonical) energy-momentum tensor. 
\end{example}

\begin{center}
  ----------
\end{center}

\subsection{Exercise 7}
\label{sec:exercise-7}

Show that $\delta_{\xi}\phi=\xi^\nu\partial_\nu\phi$ (i) is a
variational symmetry of the Klein-Gordon Lagrangian if (i) $\xi^\nu$ is a
Killing vector of flat space; (ii) $\xi^\nu$ is a conformal Killing vector in 2
spacetime dimensions in the massless case. Show that
$\delta_\xi A^a_\mu =\xi^\nu\partial_\nu A_\mu^a+\partial_\mu \xi^\nu
A_\nu^a$ is a variational symmetry of the Yang-Mills Lagrangian if (i)
$\xi^\nu$ is a Killing vector of flat space; (ii) $\xi^\nu$ is a
conformal Killing vector of flat space in 4 spacetime dimensions.

Hint: Contract the equation defining the energy-momentum tensor by
$\xi^\nu$ and conclude using the properties of the latter.

\begin{center}
  ----------
\end{center}

An
important property of evolutionary vector fields is that they form an
algebra under the commutator of vector vector fields,
\begin{equation}
\begin{split}
  \left[\delta_{Q_1},\delta_{Q_2} \right] &
  =\partial_{(\mu)}\left(\delta_{Q_1}Q_2^i-\delta_{Q_2}Q_1^i \right)
  \frac{\partial}{\partial\phi^i_{(\mu)}}\\
 & \equiv \partial_{(\mu)}[Q_1,Q_2]^i\frac{\partial}{\partial\phi^i_{(\mu)}},
\end{split}
\end{equation}
where we have used \eqref{eq:3variationder} repeatedly. 

Both equations of motion symmetries and variational symmetries form a
sub-algebra thereof. The latter follows because
\begin{equation}
\begin{split}
  [\delta_{Q_1},\delta_{Q_2}]L & =\delta_{Q_1}\partial_\mu k^\mu_{Q_2}
  -\delta_{Q_2}\partial_\mu k^\mu_{Q_1}  \\
 & =\partial_\mu\left(\delta_{Q_1} k^\mu_{Q_2} -\delta_{Q_2}
   k^\mu_{Q_1} \right) . 
\end{split}
\end{equation}

\begin{center}
  ----------
\end{center}

\subsection{Exercise 8}
\label{sec:exercise-8}

Let $j^\mu_Q$ be the Noether current associated to a variational
symmetry $Q^i$ with $j_Q=j^\mu (d^{n-1}x)_\mu$ the associated $n-1$
form. Define the (covariant) Dickey bracket through
\begin{equation}
  \label{eq:30}
  \{j_{Q_1},j_{Q_2}\}=\delta_{Q_1} j_{Q_2}.
\end{equation}
Prove that
\begin{equation}
  \label{eq:31}
  \{j_{Q_1},j_{Q_2}\}=j_{[Q_1,Q_2]}+{\rm trivial},
\end{equation}
where trivial includes constants in spacetime dimension $1$.
Hint: apply $\delta_{Q_1}$ to $dj_{Q_2}=Q^i_2\frac{\delta
  L}{\delta \phi^i}d^nx$ and use \eqref{314}.

\begin{center}
  ----------
\end{center}

\section{Gauge Symmetries}

A gauge symmetry
corresponds to an (infinite-dimensional) sub-set of variational
symmetries that depend on an arbitrary local function $f$ and its
derivatives. It is defined by
\begin{align}
\delta_f\phi^i=Q^i(f)=Q^{i(\mu)}\partial_{(\mu)}f, && \delta_f L=\partial_\mu k^\mu_f.
\end{align}
The most simple example that comes to mind is that of electromagnetism, where
\begin{equation}
\delta_f A_\mu =\partial_\mu f.
\end{equation}
Gauge symmetries can also be characterized through the important
\begin{theorem}(Noether's second theorem) There is a one-to-one
  correspondence between Noether identities and gauge symmetries.
\begin{proof}
  Let us first prove that we can associate a Noether identity to a
  gauge symmetry. Indeed, we have
\begin{equation}\label{eq:3Noether2}
Q^i(f)\frac{\delta
  L}{\delta\phi^i}
=\partial_\mu j^\mu_f.
\end{equation}
By doing multiple integration by parts, we can rewrite this equation
as
\begin{equation}
(-)^{|\mu|}f\partial_{(\mu)}\left[Q^{i(\mu)}\frac{\delta
    L}{\delta\phi^i}
\right]\equiv f Q^{\dagger i}\left[\frac{\delta
    L}{\delta\phi^i}\right]
=\partial_\mu\left(j^\mu_f-t^\mu_f \right),\label{eq:4Noether2}
\end{equation}
where $t^\mu_f$ vanishes on-shell, and we have defined the
\textbf{Noether operator} $Q^{\dagger i} $, which is the formal
adjoint of $Q^{i(\mu)}\partial_{(\mu)}$. In particular, since $f$ is
arbitrary, we can replace it by a new independent field on the
jet-space. This allows one to take the Euler-Lagrange derivative with
respect to $f$. Since the
Euler-Lagrange derivative annihilates total derivatives, we thus
arrive at the Noether identity
\begin{equation}
Q^{\dagger i}\left(\frac{\delta L}{\delta\phi^i} \right)=0.\label{eq:NI}
\end{equation}

To prove the converse, let us start by multiplying a Noether identity
by an arbitrary local function $f$, 
\begin{equation}
fN^i\left(\frac{\delta L}{\delta\phi^i} \right)=0.
\end{equation}
It is then again only a matter of multiple integrations by parts to
arrive at an equation of the form \eqref{eq:3Noether2}, with
$Q^i(f)=N^{\dagger i}(f)$, and $j^\mu_f=t^{\prime\mu}_f\approx 0$. 
\end{proof}
\end{theorem}
\begin{remark}
Let us note also that, when combining \eqref{eq:4Noether2} with the
Noether identity \eqref{eq:NI}, 
\begin{align}
\partial_\mu\left(j^\mu_f-t^\mu_f \right)=0, && \implies && d_H(j_f-t_f)=0,
\end{align}
so that $j_f-t_f$ is a closed $n-1$-form.
However, we know from Theorem \ref{th:1cohomo} that
\begin{equation}
H^{n-1}(d_H)=\delta^{n-1}_0\mathbb R.
\end{equation}
This means that
\begin{equation}
j_f=t_f+d_H\eta^{(n-2)}+\delta^{n-1}_0C,
\end{equation}
which is the sum of a term which vanishes on-shell and an exact term,
up to a constant $C$ in 1 spacetime dimension, i.e., in classical
mechanics. We thus obtain as a corollary the result that \textbf{the
  Noether current associated to a gauge symmetry is always trivial}.
\end{remark}
\begin{remark}
  Because all theories, including for instance scalar field theories,
  admit trivial Noether identities, they also admit gauge
  symmetries. Note however that to a trivial Noether identity
  corresponds a trivial gauge symmetry in the sense that the
  characteristic $Q^i[f]\approx 0$. More generally, again as a
  corrollary, any gauge symmetry in an irreducible gauge theory can be
  written in terms of a generating set
  as
  \begin{equation}
    \label{eq:18}
    Q^i(f)=R^i_\alpha\left(Z^\alpha(f)\right)+(-)^{|\mu|}\partial_{(\mu)} \left(
    M^{[j(\nu)i(\mu)]}\partial_{(\nu)} \frac{\delta L}{\delta
      \phi^j}f\right). 
  \end{equation}
\end{remark}

\begin{center}
  ----------
\end{center}

\subsection{Exercise 9}
\label{sec:exercise-9}

Study the global symmetries and the conserved currents of the first
order Hamiltonian action
\begin{equation}
  \label{eq:32}
  S_H=\int dt [p_i\dot q^i-H(q,p,t)].
\end{equation}
Hint: Use the fact that one can show that every symmetry that vanishes on-shell is
trivial, i.e., is an antisymmetric combination of the equations of
motion and thus a trivial gauge symmetry. The Hamiltonian
equations of motion can then be used to remove all time derivatives in
the characteristic of global symmetries, so that one may assume
\begin{equation}
  \label{eq:33}
  \delta q^i=Q^i(q,p,t),\quad \delta p_i=P_i(q,p,t). 
\end{equation}
With this starting point, write the out the condition that such
symmetry is a variational symmetry in Noether form
\begin{equation}
  \label{eq:34}
  Z^A\frac{\delta L_H}{\delta z^A}=\frac{d}{dt} j,\quad  z^A=(q^i,p_i),
\end{equation}
by identifying time-derivatives to conclude that, up to constants,
conserved currents are determined by $j(z^A,t)$, such that
\begin{equation}
  \label{eq:36}
  \frac{\partial}{\partial t} j+\{j,H\}=0,
\end{equation}
and the associated global symmetries are given by the Hamitlonian
vector field generated by $j$, 
\begin{equation}
  \label{eq:35}
  Z^A=\omega^{BA}\frac{\partial j}{\partial z^B},\quad
  \omega^{AB}=\begin{pmatrix} 0 & \delta^i_l\\ -\delta^k_j & 0 
  \end{pmatrix}
\end{equation}
Show that the associated Dickey bracket is the Poisson bracket.

\begin{center}
  ----------
\end{center}

\section{Generalized Noether theorems}
\label{sec:gener-noeth-theor}

The tools developed so far also allow us to give a more complete
version of Noether's first theorem than can be found in Noether's
original paper or in field theory
textbooks. 
\begin{theorem}(Noether's first theorem)
  There is a one-to-one correspondence between non-trivial
  Noether symmetries and non-trivial Noether currents, which we can write as
\begin{equation}
[Q^i]\longleftrightarrow[j^{n-1}],
\end{equation}
where the equivalence classes $[Q^i]$ are defined by
\begin{equation}
  Q^i\sim Q^i+R{^i}_\alpha(f^\alpha)+(-)^{|\mu|}\partial_{(\mu)}
  \left[M^{[j(\nu)i(\mu)]}\partial_{(\nu)}\frac{\delta L}{\delta\phi^i} \right],
\end{equation}
while $j^{n-1}\sim j^{n-1}+ t^{n-1}+d_H \eta^{n-2}+\delta^{n-1} C$
with $t^{n-2}\approx 0$. 
\end{theorem}
\begin{proof}
The proof follows by spelling out the details of Theorem \ref{th4} in
  form degree $p=n-1$,
\begin{equation}
H_1^n(\delta|d_H)\simeq H_0^{n-1}(d_H|\delta)/\delta^{n-1}\mathbb R. 
\end{equation}
Indeed, the RHS has already been treated in the context of
characteristic cohomology and its Koszul-Tate resolution. For the
LHS, we note that in maximal form degree $n$, by uniquely fixing the
$d_H$-exact term in the coboundary condition, canonical
representatives in antifield number $1$ and $2$ are given by
\begin{equation}
  \label{eq:19}
  \begin{split}
  & \omega^n_1=\phi^*_i Q^i d^nx,\\ & \omega_2^n=\left(f^\alpha C^*_\alpha
  +\frac{1}{2}(-)^{|\mu|}\partial_{(\mu)}[
  M^{j(\nu)[i(\mu)]}\phi^*_{j(\nu)}]\phi^*_i\right)d^nx, 
\end{split}
\end{equation}
with $Q^i,f^\alpha, M^{[j(\nu)i(\mu)]}$ local functions. 
For such a representative in antifield number $1$,
the cocycle condition $\delta\omega_1^n+d_H\omega_0^{n-1}=0$ reduces 
to the condition that $Q^i$ defines a variational symmetry,
\begin{equation}
  \label{eq:20}
  Q^i\frac{\delta L}{\delta \phi^i}=\partial_\mu j^\mu.
\end{equation}
For the coboundary condition $\omega_1^n=\delta\eta_2^n+d_H(\cdot)$,
one may assume that $\eta_2^n$ is of the form of $\omega_2^n$ above by
suitably adjusting the $d_H$-exact term. Taking an Euler-Lagrange
derivative of the resulting equation with respect to $\phi^*_i$ then
leads to
$Q^i=R^i_\alpha(f^\alpha)+(-)^{|\mu|}\partial_{(\mu)}[
M^{[i(\mu)j(\nu)]}\frac{\delta L}{\delta\phi^j}]$.
\end{proof}
\begin{remark}
  From the viewpoint of Noether's first theorem, all gauge symmetries
  should be considered as trivial; physically distinct ``global''
  symmetries are thus described by equivalence classes of variational
  symetries, up to gauge symmetries.
\end{remark}

Under suitable assumptions, one may then also show
that
\begin{equation}
H_p^n(\delta|d_H)=0,\quad p>2,\label{eq:22}
\end{equation}
in irreducible gauge theories. Theorem \ref{th4} then allows one to
conclude that there is no characteristic cohomology, up to
constants in spacetime dimensions $p$.

There remains the case of $p=2$,
\begin{equation}
H_2^n(\delta|d_H)\simeq H^{n-2}_0(d_H|\delta)/\delta^{n-2}_0\mathbb R, 
\end{equation}
By a similar, but more involved, reasoning to that used for $p=1$, it
can be shown that the computation of characteristic cohomology in
degree $n-2$ (up to constants in spacetime dimensions $2$) reduces to
the problem of finding physically distinct {\bf global reducibility
  parameters},
that is to say
equivalence classes of $[f^\alpha]$'s such that 
\begin{equation}
  \label{eq:21}
  R^i_\alpha(f^\alpha)\approx 0,\quad f^\alpha\sim
  f^\alpha+t^\alpha,\quad t^\alpha\approx 0.  
\end{equation}

In a large class of theories, such as Yang-Mills theories or general
relativity in spacetime dimensions greater or equal to three, one may
show that the equivalence classes of $f^\alpha$'s are determined by
field independent ordinary functions $\bar f^\alpha(x)$ satisfying the
strong equality
\begin{equation}
  R^i_\alpha(\bar f^\alpha)=0. \label{eq:3gaugepar}
\end{equation}
In this case, like in the case of Noether's first theorem, there is an
explicit formula for the associated ``surface
charges'',
that is to say representatives
for the corresponding characteristic cohomology. Indeed, by repeated
integrations by part, one has
\begin{equation}
  R^i_\alpha(f^\alpha)\frac{\delta L}{\delta\phi^i}
  =f^\alpha R^{\dagger i}_\alpha\left(\frac{\delta L}{\delta\phi^i}
  \right)+\partial_\mu S^{\mu i}_f\left(\frac{\delta
    L}{\delta\phi^i}\right). \label{ncg} 
\end{equation}
The first term on the right-hand side is zero because of the Noether
identities. This means that
$S^{\mu i}_f(\frac{\delta L}{\delta\phi^i})$ represent the weakly
vanishing Noether currents associated to the gauge symmetries
$R^i_\alpha(f^\alpha)$, that can readily be worked out by keeping the
boundary terms in \eqref{ncg}.  When using field independent
reducibility parameters $\bar f^\alpha$ satisfying
\eqref{eq:3gaugepar}, the right hand side vanishes as well. It follows
that the $n-1$ form
$S_{\bar f}=S^{\mu i}_f[\frac{\delta L}{\delta\phi^i}](d^{n-1}x)_\mu$
satisfies
\begin{equation}
d_H S_{\bar f}=0. 
\end{equation}
Because the horizontal cohomology of degree $n-1$ is trivial
(cf.~Theorem \ref{th:1cohomo}), this means that
\begin{equation}
(0\approx) S_{\bar{f}}=d_Hk_{\bar{f}},\quad k_{\bar f}=\rho_H(S_{\bar
  f}). 
\end{equation}
In other words, all linearly independent surface charges can be explicitly
constructed by applying the contracting homotopy of the horizontal
complex to the weakly vanishing Noether current associated with
linearly independent solutions of \eqref{eq:3gaugepar}.

In the general case, where one has a basis $[f_A^\alpha]$ of possibly
field dependent global reducibility parameters, one may show that
$\rho_H(S_{f_A})$ still provide a basis of characteristic cohomology
in degree $n-2$ (up to constants in spacetime dimension $2$).

Working out explicit expressions for the surface charges is direct but
quite lenghty in the case of second order field equations, that is why
we will not do so here.
\begin{example}
  Let us give some examples of field-independent reducibility
  parameters. In the case of electromagnetism, where
\begin{equation}
\delta_f A_\mu=\partial_\mu f, 
\end{equation}
we find
\begin{align}
\partial_\mu\bar{f}=0, && \Longrightarrow && \bar{f}=const,
\end{align}
which characterizes the electric charge.

For the case of Yang-Mills and Chern-Simons theory based on a
sem-simple Lie group, and of
General Relativity, condition \eqref{eq:3gaugepar} reads
\begin{align}
D_\mu \bar{f}=0, && \mathcal{L}_{\bar{\xi}}g_{\mu\nu}=0,
\end{align}
both of which admit only the trivial solutions
$\bar{f}^\alpha=\bar{\xi}^\mu=0$ because the equations need to hold for
arbitrary gauge potentials or metrics.

In the case of Yang-Mills theory or gravity linearized around a
background solution $\bar g_{\mu\nu}$ or $\bar A$, a generating set of
gauge transformations is given by
\begin{align}
\delta_f a_\mu =D^{\bar A}_\mu
f, && \delta h_{\mu\nu}= \mathcal{L}_{{\xi}} \bar
g_{\mu\nu}\label{eq:23}. 
\end{align}
In case the background solution is flat for instance, $\bar
A=g^{-1}dg$, $\bar g_{\mu\nu}=\eta_{\mu\nu}$, the solutions to
equation \eqref{eq:3gaugepar} are given by
\begin{equation}
  \label{eq:24}
  \bar f=\lambda^\alpha g^{-1} T_\alpha g,\quad \bar
  \xi_\mu=a_\mu+\omega_{[\mu\nu]} x^\nu, 
\end{equation}
and the associated surface charges are related to color charges,
respectively the ADM charges in general relativity
\cite{Abbott:1981ff,Abbott:1982jh}.
\end{example}
\begin{remark}
  The above results on surface charges associated to global
  reducibility parameters can be extended to the case of {\bf
    asymptotic symmetries} and the associated current algebras.
\end{remark}

\section{Gauge symmetry algebra}
\label{sec:gauge-symm-algebra}

From their definitions, it follows with little work that gauge
symmetries form an ideal in the algebra of variational symmetries,
while on-shell vanishing gauge symmetries in turn form an ideal in the
sub-algebra of gauge symmetries.
The latter may thus be quotiented away and the
resulting quotient is the algebra of non-trivial gauge symmetries we
are interested in.

The information on non-trivial gauge symmetries is contained in the
generating set. We may start by considering gauge symmetries of the
form $\delta_\epsilon\phi^i=R^i_\alpha(\epsilon^\alpha)$ with
$\epsilon^\alpha(x)$ arbitrary functions. The commutator
$[\delta_{\epsilon_1},\delta_{\epsilon_2}]$ is a variational symmetry
that depends on arbitrary functions. Up to on-shell vanishing gauge
symmetries, it may thus be written in terms of the generating set,
\begin{equation}\label{eq:3generating}
  \delta_{\e_1}R^i_\alpha(\e_2^\alpha)-\delta_{\e_2}R^i_\alpha(\e_1^\alpha)
  \approx R^i_\gamma\left(f^\gamma_{\a\b}(\e^\alpha_1,\e^\beta_2)\right),
\end{equation}
for some bi-differential, skew-symmetric operators ``structure operators''
\begin{equation}
f^{\gamma(\mu)((\nu)}_{\a\b}\partial_{(\mu)}\epsilon^\a_1
\partial_{(\nu)}\epsilon^\b_2\label{eq:25}.
\end{equation}
From the identity
\begin{equation}
\left[\delta_{\e_1},[\delta_{\e_2},\delta_{\e_3}]\right]+{\rm cyclic}
(1,2,3)=0,\label{eq:53}
\end{equation}
where ``cyclic'' stands for cyclic permutations, and the
irreducibility assumption for the generating set, it then follows that
the structure operators satisfy the generalized Jacobi identities,
\begin{equation}\label{eq:3Jacobi}
  \delta_{\e_1}f^\c_{\a\b}(\e_2^\a,\e_3^\b)+
  f^\c_{\d\rho}(\e_1^\d,f^\rho_{\a\b}(\e_2^\a,\e_3^\b))+{\rm
  cyclic}(1,2,3)\approx 
0. 
\end{equation}
This algebraic structure of vector fields in involution may be
captured through a suitable differential.
Instead of the arbitrary functions
$\epsilon^\alpha(x)$ one introduces so-called ghosts, that is to say
Grassmann odd fields $C^\alpha(x)$, that are promoted to additional
coordinates in the fiber of the jet-bundle. Defining the action on
generators as
\begin{equation}
  \label{eq:26}
  \gamma\phi^i =R^i_\alpha(C^\alpha),\quad \gamma
  C^\gamma=-\frac{1}{2} f^\gamma_{\alpha\beta}(C^\alpha,C^\beta),
\end{equation}
and extending to the jet-bundle as 
\begin{equation}
  \label{eq:27}
  \gamma=\partial_{(\mu)}\left(R^i_\alpha(C^\alpha)\right)\frac{\partial}{\partial
    \phi^i_{(\mu)}}- \frac{1}{2}
  \partial_{(\mu)}\left(f^\gamma_{\alpha\beta}(C^\alpha,C^\beta)\right)
  \frac{\partial}{\partial
    C^\alpha_{(\mu)}}, 
\end{equation}
it follows from \eqref{eq:3generating} and \eqref{eq:3Jacobi} that
\begin{equation}
  \label{eq:28}
  \gamma^2\approx 0. 
\end{equation}
\begin{remark}
  When using general non-trivial gauge symmetries with arbitrary local
  function $f^\alpha$ instead of arbitrary functions
  $\epsilon^\alpha$, equation \eqref{eq:3generating} becomes
  \begin{equation}
    \begin{split}
    & \delta_{f_1}R^i_\alpha(f_2^\alpha)-\delta_{f_2}R^i_\alpha(f_1^\alpha)
  \approx
  R^i_\gamma\left([f_1,f_2]^\gamma_A\right),\\ & [f_1,f_2]^\gamma_A
  =f^\gamma_{\a\b}(f^\alpha_1,f^\beta_2)+\delta_{f_1}f^\gamma_2
  -\delta_{f_2}f^\gamma_3, \
\end{split}
\end{equation}
while \eqref{eq:3Jacobi} turns into the statement that, on-shell, the
bracket $[f_1,f_2]_A$ satisfies the Jacobi identity. On-shell, the
algebraic structure that emerges in this way for irreducible gauge
theories is that of a Lie algebroid (compare for instance to section
2.1 of \cite{LojaFernandes2006}).
\end{remark}

\begin{center}
  ----------
\end{center}

\subsection{Exercise 10}
\label{sec:exercise-10}

Show that 
\begin{equation}
  \delta_\epsilon X^i=-  \alpha^{ij}\epsilon_j,\quad
\delta_\epsilon \eta_{i\mu}=\partial_\mu \epsilon_i+\partial_i
\alpha^{jk}\eta_{j\mu}\epsilon_k,\label{eq:38}
\end{equation}
are gauge symmetries of the Poisson-Sigma model. Work out the
structure functions and show that the gauge algebra is open. Work out
the coefficient of the on-shell vanishing terms. 

\begin{center}
  ----------
\end{center}

\chapter{Batalin-Vilkovisky formalism} 
\label{ch:Lec4}

The (original) aim of the \textbf{Batalin-Vilkovisky} (BV)
\textbf{formalism} is to control gauge invariance during the
perturbative quantization of gauge theories. It builds on the methods
of Faddeev and Popov, Slavnov, Taylor, Zinn-Justin,
Becchi-Rouet-Stora, and Tyutin designed for Yang-Mills type theories
with closed algebras involving structure constants (see
\cite{Piguet:1995er} for a thorough review), and extends them to
generic gauge theories with field-dependent structure operators and
open algebras. Besides \cite{Henneaux:1992ig} and the references in
the introduction of \cite{Barnich:2000zw}, useful reviews include
\cite{Gomis:1995he} and \cite{Weinberg:1996kr}, section 15.9.

\section{Motivation}

Up to now
we have dealt with classical gauge field theories. If one is
interested in the perturbative quantization of field theories, the
main formula, encoding the Feynman rules for Green's functions, is 
\begin{equation}
  \frac{\mathcal{Z}[J]}{\mathcal{Z}_0[0]}  =\frac{\int
    [d\phi]e^{\frac{i}{\hbar}\left(S_0[\phi]+S_I[\phi]+J_A\phi^A
      \right)}}
  {\int[d\phi]e^{\frac{i}{\hbar}S_0[\phi]}} 
 = e^{\frac{i}{\hbar}S_I\left[\frac{\hbar}{i}\frac{\delta}{\delta J}
   \right]}e^{\frac{i}{2\hbar}J_A(\mathcal{D}^{-1})^{AB}J_B}
\end{equation}
Here
\begin{equation}
  S_0=-\frac{1}{2} \phi^A {\mathcal D}_{AB}\phi^B,
\end{equation}
is the quadratic part of the Larangian, including
$i\epsilon$ terms so that a unique inverse $(\mathcal{D}^{-1})^{AB}$
exists, while $S_I$ encodes cubic and higher order vertices. In the
case of gauge theories, the non-trivial aspects we have treated at the
classical level have direct counterparts at the quantum level: a
consequence of non-trivial gauge symmetries and Noether identities is
that the quadratic kernel of the action is no longer invertible.
\begin{remark}
  Note that in this notation (due to DeWitt), the index $A=(i,x^\mu)$
  includes the spacetime (or momentum) dependence, summation over $A$
  then includes a spacetime (or momentum) integral and $\delta^A_B$
  includes Dirac deltas.
\end{remark}

\begin{center}
  ----------
\end{center}

\subsection{Exercise 11}
\label{sec:exercise-11}

In the case of the Lagrangian for free electromagnetism, show that the
non-invertibility of the quadratic kernel in momentum space is
directly related to gauge invariance. Hint: What is the eigenvector
of eigenvalue zero of the quadratic kernel ?

\begin{center}
  ----------
\end{center}

This implies that we cannot directly define a propagator.  In order to
perform perturbative calculations, we then have to gauge-fix the
system in such a way that the quadratic kernel becomes
invertible. However, if one gauge-fixes the theory in a pedestrian way
one looses all information about the original gauge invariance of the
system: how can we be sure then that a different gauge-fixing would
not give us different physical results?

The (original) aim of the \textbf{Batalin-Vilkovisky} (BV)
\textbf{formalism} is to make the quadratic kernel invertible while
retaining the consequences of gauge invariance. 

\section{BV antibracket, master action and BRST differential}
\label{sec:bv-antibr-mast}

The core of the formalism
is the introduction of an anti-canonical structure. This is possible
because during our classical treatment we have introduced, besides the
original fields $\phi^i$, ghosts $C^\alpha$ for an irreducible
generating set of gauge symmetries. This extended set of fields is
denoted by $\Phi^a=(\phi^i,C^\alpha)$. Associated antifields
$\Phi^*_a=(\phi^*_i,C^*_\alpha)$ have been introduced for the
Koszul-Tate resolution of the equations of motion. It is then possible
to consider fields and antifields as conjugated variables through an
odd graded Lie bracket called the \textbf{antibracket}. The
antibracket is defined, for two arbitrary functionals $F,G$, as
\begin{equation}
(F,G)=\int
d^nx\,\left(\frac{\delta^RF}{\delta\Phi^a(x)}
  \frac{\delta^LG}{\delta\Phi_a^*(x)}
-\frac{\delta^RF}{\delta\Phi_a^*(x)}\frac{\delta^LG}{\delta\Phi^a(x)}\right),
\end{equation}
where the L/R superscripts on the functional derivatives denote that
they are taken respectively from the left or from the right. Remember
that when dealing with anticommuting variables right and left
derivatives can differ by sign factors. Remember also that in the case
of a local functional
\begin{equation}
F=\int d^nx\, f,
\end{equation}
the functional derivative is given by the Euler-Lagrange derivative of
its integrand
\begin{equation}
  \frac{\delta^RF}{\delta\Phi^a(x)}
  =\frac{\delta^Rf}{\delta\Phi^a}\biggr\rvert_{\phi(x)},
\end{equation}
evaluated at a section, with similar relations for antifields. With
respect to the $\mathbb{Z}$-grading, called \textbf{ghost number},
\begin{center}
\begin{tabular}{ c | c c c c }
 gh & 0 & 1 & -1 & -2 \\ \hline 
 & $\phi^i$ & $C^\alpha$ & $\phi_i^*$ & $C_\alpha^*$ \\    
\end{tabular},
\end{center}
the antibracket has ghost number 1, in the sense that if
$\mathcal{F}^{g_i}$ is a functional of ghost number $g_1$, then
\begin{equation}
  (\cdot,\cdot):\mathcal{F}^{g_1}\times\mathcal{F}^{g_2}
  \longrightarrow\mathcal{F}^{g_1+g_2+1}.
\end{equation}

The antibracket
has the following properties:

\noindent {\bf graded antisymmetry:} 
\begin{equation}
(F,G)=-(-)^{(|F|+1)(|G|+1)}(G,F),\label{eq:29}
\end{equation}
{\bf graded Jacobi identity:}
\begin{equation}
(F,(G,H))=((F,G),H)+(-)^{(|F|+1)(|G|+1)}(G,(F,H))\label{eq:37},
\end{equation}
where $|G|$ denotes the $\mathbb{Z}_2$-grading of $G$, which is 1 for
fermionic quantities and 0 for bosonic ones.

Note also that for a bosonic functional $B$, one can easily show that
\begin{equation}
  \frac{1}{2}(B,B)=\int d^nx\,
  \frac{\delta^RB}{\delta\Phi^a(x)}\frac{\delta^LB}{\delta\Phi_a^*(x)}
  =-\int d^nx\,\frac{\delta^RB}{\delta\Phi_a^*(x)}\frac{\delta^LB}{\delta\Phi^a(x)},
\end{equation}
when using that
\begin{equation}
  \label{eq:39}
  \frac{\delta^RF}{\delta\Phi^a(x)}=(-)^{|a|(|a|+|F|)}
  \frac{\delta^LF}{\delta\Phi^a(x)}, 
 \end{equation}
 and similar relations for antifields, with $|\phi^a|=|a|$ and
 $|\Phi_a^*|=|a|+1$ and ${\rm gh}\,{\Phi^*_a}=-{\rm gh}\,{\phi^a}-1$.

 For an irreducible gauge theory, the BV master action
 is a functional of ghost number $0$ that starts with
 the classical action to which one couples though antifields an
 irreducible generating set of gauge transformations with gauge
 parameters replaced by ghosts,
\begin{equation}\label{eq:3BVMaster1}
S=\int d^nx\left[\mathcal{L}+\phi_i^*R^i_\alpha(C^\alpha)+\dots \right],
\end{equation}
where $\mathcal L$ is the classical Lagrangian density. The higher
order terms hidden in $\dots$ are completely determined by requiring
that $S$ satisfies the \textbf{master equation}
\begin{equation}
\frac{1}{2} (S,S)=0.
\end{equation}
The (antifield-dependent) {\bf BRST differential}
is canonically generated through the master action (in the same sense
as in classical mechanics the Hamiltonian is the canonical generator
of time translations in the Poisson bracket),
\begin{equation}
s= (S,\cdot).\label{eq:gen}
\end{equation}
It raises the ghost number by $1$ and is nilpotent, $s^2=0$. Indeed,
the graded Jacoby identity implies that
\begin{equation}
(S,(S,\cdot))=((S,S),\cdot)-(S,(S,\cdot))=-(S,(S,\cdot)),
\end{equation}
when using the master equation.

With these ingredients, it can be proven that
\begin{theorem}
  The solution to the master equation exists, it is uniquely defined
  up to anticanonical transformations\footnote{Anticanonical
    transformations are the obvious generalizations of the canonical
    transformations of Hamiltonian mechanics.}, and the solution is a
  local functional.
\end{theorem}
An important additional degree to prove this theorem on the existence
of the master action and to analyse the antifield dependent BRST
differential is the antifield number.
As we have seen before, for irreducible gauge theories, it assigns $1$
to the $\phi^*_i$ and their derivatives, $2$ to the $C^*_\alpha$ and
their derivatives, and $0$ to all other variables. It then follows
that the classical action is of antifield number $0$, the second term
in \eqref{eq:3BVMaster1} is of antifield number $1$. More precisely,
the theorem states that the $\dots$ in \eqref{eq:3BVMaster1} are terms
of antifield number higher or equal to $2$. Accordingly, the expansion
of BRST differential in terms of the antifield number is
\begin{equation}
  \label{eq:42}
  s=\delta +\gamma +\sum_{k\geq 0} s_k,
\end{equation}
where $\delta$ lowers the antifield number by $1$, $\gamma$ preserves
the antifield number, while $s_k$ raises the antifield number by $k$. 

From \eqref{eq:gen}, it follows that the action of the BRST
differential on fields and antifields is explicitly given by
\begin{equation}
\begin{cases}
s\Phi^a(x)=-\frac{\delta^RS}{\delta\Phi^*_a(x)}, \\
s\Phi^*_a(x)=\frac{\delta^RS}{\delta\Phi^a(x)}.
\end{cases}
\end{equation}
Since $S$ is a local functional, the BRST differential can be extended
to the derivatives of the fields and antifields and written as a
generalized vector field on the jet-bundle that commutes with total
derivatives.

It also follows from \eqref{eq:3BVMaster1} and the definition of the
antifield number that $\delta$ coincides with the Koszul-Tate
differential analysed previously,
\begin{equation}
  \label{eq:43}
  \delta \phi^*_i=\frac{\delta L}{\delta\phi^i},\quad \delta
  C^*_\alpha = R^{+i}_\alpha[\phi^*_i],\quad \delta \Phi^a=0, 
\end{equation}
while the BRST differential encodes gauge invariance in the sense
that
\begin{equation}
  \label{eq:44}
  \gamma \phi^i = R^i_\alpha(C^\alpha), 
\end{equation}
i.e., to lowest order in antifield number, the BRST transformation of
the original fields is a gauge transformation with gauge parameters
replaced by anticommuting ghosts.

The higher order terms in the master action encode how complicated the
algebra of gauge symmetries actually is. For instance, the next term
is given in terms of the structure operators of the algebra,
\begin{equation}
\int d^nx \frac{1}{2} C^*_\alpha
f^\alpha_{\beta\gamma}(C^\beta,C^\gamma).\label{eq:40}
\end{equation}
This terms determines the BRST transformations of the ghosts to lowest
order in antifield number, 
\begin{equation} \label{eq:4BRSTghost1}
\gamma C^\alpha=-\frac{1}{2}f{^\alpha}_{\beta\gamma}(C^\beta,C^\gamma).
\end{equation}

For theories with closed algebras involving structure constants or
field-independent structure operators, this is the only additionl term
needed in the master action.
\begin{center}
  ----------
\end{center}

\subsection{Exercise 12}
\label{sec:exercise-12}

Check that in this particular case, the expansion of the BRST
differential according to antifield number stops in degree $0$,
$s=\delta+\gamma$, and work out $\gamma\phi^*_i$, $\gamma C^*_\alpha$.

\begin{center}
  ----------
\end{center}

The
first, simple example
is as usual electromagnetism, for which
\begin{align}
\phi^i=A_\mu, && R_\mu(C)=\partial_\mu C.
\end{align}
In this case, no additional nonlinear terms have to be added to the
master action, which then is
\begin{equation}
S=\int d^4x\left[-\frac{1}{4}F_{\mu\nu}F^{\mu\nu}+A^{*\mu}\partial_\mu C \right].
\end{equation}

For \textbf{nonabelian gauge theories} based on semisimple Lie
algebras, a case which includes both Yang-Mills and Chern-Simons
theories, the only additional nonlinear term required is the one
involving the structure constants,
\begin{equation}\label{eq:3ActionGauge1}
  S=S_{YM/CS}[A]+\int d^nx\left[A^{*\mu}_\alpha D_\mu C^\alpha+
    \frac{1}{2}C^*_\alpha f{^\alpha}_{\beta\gamma}C^\beta C^\gamma \right].
\end{equation}
\begin{remark}
  For electromagnetism, Yang-Mills and Chern-Simons theories,
  couplings to matter fields $y^i$, which can be (complex) scalars,
  Dirac or Weyl fermions that transform under a matrix resesentation
  representation ${{T_\alpha}^i}_j$ of the gauge algebra,
\begin{equation}
  \label{eq:45}
  [T_\alpha,T_\beta]=f^\gamma_{\alpha\beta}T_\gamma,
\end{equation}
are introduced as follows. In the Lagrangian $L_M[y,\partial y]$,
supposed to be invariant under
$\delta_k y^i=-k^\alpha {{T_\alpha}^i}_j y^j$ with $k^\alpha$
constant, one replaces $\partial_\mu y^i$ by
$D_\mu y^i=\partial_\mu y^i + A^\alpha_\mu {{T_\alpha}^i}_j y^j$. It
the follows that $L_M[y,Dy]$ is invariant under the gauge
transformations $\delta_\epsilon A^\alpha_\mu$ and
$\delta_\epsilon y^i=-\epsilon^\alpha {{T_\alpha}^i}_j y^j$ with
spacetime dependent gauge parameters $\epsilon^\alpha(x)$. The master
action for the complete theory is then the one for electromagnetism,
Yang-Mills and Chern-Simons theories to which one adds
\begin{equation}
  \label{eq:46}
  \int d^nx \left(L_M[y,Dy]-C^\alpha {{T_\alpha}^i}_j y^j
    y^*_i\right). 
\end{equation}
\end{remark}

For \textbf{General Relativity}, one introduces diffeomorphism
ghosts $\xi^\mu$, and instead of structure constants, one deals
with field-independent structure operators, since
\begin{equation}
\delta_\xi g_{\mu\nu}=-\mathcal{L}_\xi g_{\mu\nu},
\end{equation}
and 
\begin{equation}
[\delta_{\xi_1},\delta_{\xi_2}]g_{\mu\nu}=-\mathcal{L}_{[\xi_1,\xi_2]}g_{\mu\nu}.,
\end{equation}
The BV master action in this case is
\begin{equation}
S=\int d^nx\left[\sqrt{|g|}R+g^{*\mu\nu}\mathcal{L}_\xi
  g_{\mu\nu}+\xi^*_\mu\partial_\nu\xi^\mu\xi^\nu \right]. 
\end{equation}

Finally for the Poisson-Sigma model, based on Exercise 10, one can
show that the BV master action is given by \cite{Cattaneo:1999fm}
\begin{multline}
  \label{eq:41}
  S=S^{PSM}+\int d^2x \Big[-
  X^*_i \alpha^{ij} C_j
    +\eta^{*i\mu}(\partial_\mu C_i +\partial_i \alpha^{jk} \eta_{j\mu}
    C_k)\\+\frac{1}{2} C^{*i}\partial_i\alpha^{jk}C_jC_k+\frac{1}{4}
    \epsilon_{\mu\nu}
    \eta^{*i\mu}\eta^{*j\nu}\partial_i\partial_j\alpha^{kl}
    C_kC_l \Big],
\end{multline} with $\epsilon^{01}=1$.

\section{Gauge fixation}
\label{sec:gauge-fixation}

We already discussed how, because of gauge symmetry, the quadratic
kernel of the action is non-invertible. In fact, this is not changed
in the master action \eqref{eq:3BVMaster1}, since we have not
introduced any additional terms quadratic in the classical fields. In
order to have an invertible quadratic kernel and thus well-defined
free propagators, allowing us to perform computations in perturbation
theory, we still need to gauge-fix the theory.

In order to do so, one usually has to introduce more fields and
antifields,
belonging to the so called \textbf{nonminimal sector}: for theories of
Yang-Mills type, these will be the fermionic \textbf{antighosts}
$\bar{C}^\alpha$ of ghost number $-1$, the bosonic \textbf{Lagrange
  multiplier} (also called \textbf{Nakanishi-Lautrup auxiliary field})
$B^\alpha$ of gost number $0$ and their respective antifields
$\bar{C}_\alpha^*,B_\alpha^* $ of respective ghost numbers $0,-1$. The
complete set of fields is then
$\Phi^a=(\phi^i,C^\alpha,\bar C^\alpha,B^\alpha)$.

The master action is extended to the non-minimal sector through 
\begin{equation}
S_{NM}=S-\int d^nx \bar{C}^*_\alpha B^\alpha.
\end{equation}
The BRST transformation of these new fields generated by $S$ are then
\begin{align}
s\bar{C}^\alpha=B^\alpha, && sB^\alpha=0, \\
s\bar{C^*_\alpha}=0, && sB^*_\alpha=-\bar{C}^*_\alpha,
\end{align}
so that it is simply a shift symmetry, not affecting the physics. From
the point of view of the cohomology of $s$,
$\bar{C}^\alpha$ and $B^\alpha$, their antifields and all their
derivatives form so-called
\textbf{contractible pairs}, which means that they do not contribute
to any of the relevant cohomology groups. They are only needed as a
convenient means to fix the gauge.

This is done through a so-called \textbf{gauge-fixing fermion}
$\Psi[\Phi]$, which is a fermionic functional of the fields alone of
ghost number $-1$ (and thus necessarily dependent on fields from the
non-minimal sector). For a large class of gauges in Yang-Mills type
theories for example, it may be choosen as
\begin{equation}\label{eq:4GFermionYM1}
  \Psi=\int d^nx \bar{C}^\alpha\left(\partial_\mu A^{\mu \beta}
    -\frac{\xi}{2}B^\beta \right)g_{\alpha\beta},
\end{equation}
where $g_{\alpha\beta}$ is the Cartan-Killing metric of the gauge
group, that can be used, together with its inverse, to lower and raise
Lie algebra indices.

The gauge-fixed action is then defined by a shift of the antifields by
a "functional gradient" term of the gauge-fixing fermion:
\begin{equation}
  S_{\Psi}[\Phi,\tilde\Phi^*]=S_{NM}\left[\Phi,\tilde \Phi^*
    +\frac{\delta\Psi}{\delta\phi}\right].
\end{equation}
In the specific case of Yang-Mills type theories in which $\Psi$ is defined by
\eqref{eq:4GFermionYM1} and the solution of the master action is linear
in antifields, the gauge-fixed action is
\begin{equation}
  S_{\Psi}[\Phi,\tilde\Phi^*]=S^{YM,CS}-s\Psi- \int d^nx\,
  s\phi^a\tilde\Phi^*_a,
\end{equation}
with
\begin{equation}
  -s\Psi=\int d^nx
  \left[-\partial_\mu\bar{C}_\alpha D^\mu C^\alpha-(\partial_\mu
    A^\mu_\alpha)B^\alpha+\frac{\xi}{2}B_\alpha B^\alpha \right].
\end{equation} 

A crucial point of the BV construction is that the gauge-fixing is a
\textbf{canonical transformation}, which means that it leaves the
canonical antibracket relations invariant. In fact, it is a canonical
transformation of the second type (following the classification of
Goldstein's textbook), for which $\Psi$ is the generator. This can be
easily seen by identifying the fields with the coordinates of
classical mechanics, and the antifields with the momenta:
\begin{align}
q\leftrightarrow\Phi, && p\leftrightarrow\Phi^*.
\end{align}
We recall that a canonical transformation $qdp=QdP+dF$
of the second type is generated by a function $F$ of the form
\begin{equation}
F=F_2(q,P)-QP,
\end{equation}
and transforms the canonical variables as
\begin{align}
p=\frac{\partial F_2}{\partial q}, && Q=\frac{\partial F_2}{\partial P}.
\end{align}
Gauge-fixing is a transformation of exactly this type, with
\begin{align}
F_2=qP+\psi(q) && \implies && 
\begin{cases}
p=P+\frac{\partial\psi}{\partial q} \\
q=Q.
\end{cases}
\end{align}

This fact has the important consequence that the gauge-fixed master
action is still a solution of the master equation in the new
variables, 
\begin{equation}
(S_{\Psi},S_{\Psi})_{\Phi,\tilde\Phi^*}=0.
\end{equation}
The gauge fixed (antifield-dependent) BRST differential is defined by
$s_\psi=(S_\Psi,\cdot)_{\Phi,\tilde\Phi*}$ and is nilpotent off-shell,
$s_\Psi^2=0$.

Note that in theories with a master action that is linear in
antifields, the BRST transformations of the fields coincide before and
after gauge fixing, $s_\Psi\Phi^a=s\Phi^a$.  In the general case,
defining the gauge fixed BRST differential without antifields by 
$s^g\Phi^a=s_\Psi\Phi^a[\Phi,\tilde\Phi^*=0]$, we have 
\begin{equation}
s^g S_\Psi[\Phi,0]=0,\label{eq:47}
\end{equation}
which follows by putting the antifields $\tilde\Phi^*_a$ to zero in
$\frac{1}{2}(S_\Psi,S_\Psi)=0$.  Except for theories with a linear
dependence in antifields, $s^g$ is in general only nilpotent when the
equations of motion defined by $S_\Psi(\Phi,0)$ hold,
$(s^g)^2\Phi^a\approx 0$. This follows by putting the antifields
$\tilde\Phi^*_a$ to zero in
$(S_\Psi,(S_\Psi,\phi^a)_{\Phi,\tilde\Phi^*})_{\Phi,\tilde\Phi^*}=0$.

In order to avoid this and to better control the original gauge
invariance of the theory, it is most convenient not to put the
antifields $\tilde\Phi^*$ to zero during renormalization. These
antifields then act as {\bf sources for the (non-linear) BRST
  transformation} whose renormalization is then controlled together
with the renormalization of the (gauge fixed) action.

\begin{center}
  ----------
\end{center}

\subsection{Exercise 13}
\label{sec:exercise-13}

Eliminate the auxiliary fields $B^\alpha$ in the gauge fixed master
action with antifields for Yang-Mills type theories. Show that
$s^g \bar C^a$ is only nilpotent on-shell. Work out the momentum space
propagators in Yang-Mills and Chern-Simons theories with and without
auxiliary $B^\alpha$ fields.

\begin{center}
  ----------
\end{center}

\section{Independence of gauge fixing}
\label{sec:fradk-vilk-theor}

The
statement that we didn't spoil the physical content of gauge
invariance after gauge-fixing is captured by the following:
\begin{theorem}[Fradkin-Vilkovisky theorem]
Expectation values of
BRST-closed operators are independent of
the choice of gauge-fixing.
\end{theorem}
Indeed, let $\Psi$ and $\Psi+\delta\Psi$ be two different gauge-fixing
fermions, and $X[\Phi,\tilde\Phi^*]$ be BRST closed, i.e., such that
\begin{equation}
(S_\Psi,X)=0.\label{eq:54}
\end{equation}
In particular, if $X$ only depends on the original
fields, this condition means that $X$ is gauge invariant. we have,
\begin{equation}
\begin{split}
  & \langle0,+\infty|T\hat{X}|0,-\infty\rangle_{\Psi+\delta\Psi}
  -\langle0,+\infty|T\hat{X}|0,-\infty\rangle_{\Psi} \\
  = &\int[d\Phi]\left(e^{\frac{i}{\hbar}S_{\Psi+\delta\Psi}}
    -e^{\frac{i}{\hbar}S_{\Psi}} \right)X \\
  = &\frac{i}{\hbar}\int[d\Phi]
  \frac{\delta^RS_{\Psi}}{\delta\tilde\Phi^*_A}
  \frac{\delta^L\delta\Psi}{\delta\Phi^A} e^{\frac{i}{\hbar}S_{\Psi}}X+O((\delta\Psi)^2) \\
  = &
  (-)^{|A|+1}\frac{i}{\hbar}\int[d\Phi]\frac{\delta^L(\delta\Psi)}{\delta\Phi^A}
  \frac{\delta^RS_{\Psi}}{\delta\tilde\Phi_A^*}e^{\frac{i}{\hbar}S_{\Psi}}X+O((\delta\Psi)^2)
  \\
  =
  &\frac{i}{\hbar}\int[d\Phi]\delta\Psi\frac{\delta^L}{\delta\Phi^A}\left[(s_\psi\Phi^A)
    e^{\frac{i}{\hbar}S_{\Psi}}X\right]+O((\delta\Psi)^2) .
\end{split}
\end{equation}
Here in the second equality we have expanded the exponential of the
action to first order in the deformation of the gauge-fixing, while in
the fourth equality we have performed an integration by parts inside
the functional integral, and we have noted that
$\frac{\delta^RS_{\Psi}}{\delta\Phi_A^*}=-s_\psi\Phi^A$. Neglecting
the contact term contribution
$\frac{\delta^L}{\delta\Phi^A}(s_\Psi\Phi^A)$ which is proportional to
$\delta(0)$
and taking into account antifield dependent BRST invariance of the
gauge fixed action in the
form 
\begin{equation}
  s_\Psi\Phi^A\frac{\delta^L S_\Psi}{\delta\phi^A}=\frac{1}{2}
(S_{\Psi}, S_{\Psi})=0,\label{eq:gaugefixedBRST}
\end{equation}
and also \eqref{eq:54}, we find that
\begin{equation}
\begin{split}
&
\langle0,+\infty|T\hat{X}|0,-\infty\rangle_{\Psi+\delta\Psi}
-\langle0,+\infty|T\hat{X}|0,-\infty\rangle_{\Psi}\\
=& \frac{i}{\hbar}\int[d\Phi]\delta\Psi
e^{\frac{i}{\hbar}S_{\Psi}}(s_\psi\phi^A)
\frac{\delta^L X}{\delta\Phi^A}+O((\delta\Psi)^2)\\
= & \frac{i}{\hbar} \int[d\Phi]\delta\Psi
e^{\frac{i}{\hbar}S_{\Psi}}\frac{\delta^R
  S_\Psi}{\delta\Phi^A}\frac{\delta^L X}{\delta \tilde\Phi^*_A}+O((\delta\Psi)^2), 
\end{split}
\end{equation}
The term of first order in $\delta\Psi$ then vanishes on account of
the Schwinger-Dyson equations of the theory.
\begin{remark}
  These are formal arguments that hold at tree level. They may be
  violated by $\hbar$-correction when taking renormalization into
  account. Some of these $\hbar$-corrections are captured in an
  elegant way by the so-called quantum BV formalism. However, this
  formalism remains formal as well unless due care is devoted to
  renormalization.
\end{remark}

\section{Slavnov-Taylor identities and Zinn-Justin equation} 

BRST invariance of the gauge fixed action (with antifields)
leads to relations among correlation functions, generalizing the
\textbf{Slavnov-Taylor (ST) identities} of non-abelian gauge
theories. They are most economically stated in terms of generating
functionals. Indeed, starting from translation invariance of the path
integral, neglecting contact terms and using
\eqref{eq:gaugefixedBRST}, we get
\begin{multline}
0 =\frac{1}{Z[J,\tilde\Phi^*]}\int[d\Phi]\frac{\delta^L}{\delta\Phi^A}\left(s_\Psi\Phi^A
  e^{\frac{i}{\hbar}(S_\Psi+J_A\Phi^A)} \right)
\\=\frac{i(-)^{|A|}J_A}{\hbar Z[J,\tilde\Phi^*]} \int[d\Phi]
s_\Psi\Phi^A
e^{\frac{i}{\hbar}(S_\Psi+J_A\cdot\Phi^A)}=\frac{i}{\hbar}(-)^{|A|}J_A\langle
\widehat{s_\Psi\Phi^A}\rangle^{J,\tilde\Phi^*}.
\end{multline}
Defining the normalized generating functional for connected Green's
functions $W[J,\tilde\Phi^*]$ through
\begin{equation}
  \label{eq:48}
  e^{\frac{i}{\hbar}
    W[J,\tilde\Phi^*]}=\frac{Z[J,\tilde\Phi^*]}{Z[0,0]}, 
\end{equation}
this equation can be written as 
\begin{equation}
  \label{eq:49}
  (-)^{|A|}J_A\frac{\delta^R\mathcal{W}}{\delta\Phi^*_A}=0. 
\end{equation}
Finally, performing the Legendre transform with respect to $J_A$,
\begin{equation}
  \label{eq:50}
  \tilde \Phi^A=\langle\Phi^A\rangle^{J,\tilde\Phi^*}=\frac{\delta^L
    W}{\delta J^A},
\end{equation}
with inverse relation
$J_A=J_A[\tilde\Phi,\tilde\Phi^*]$, 
we can define the \textbf{quantum effective action},
\begin{equation}
  \Gamma[\tilde{\Phi},\tilde\Phi^*]
  =\left[\mathcal{W}[J,\tilde\Phi^*]-J_A\tilde{\Phi}^A \right]
  \bigg\rvert_{J=J[\tilde{\Phi},\tilde\Phi^*]},
\end{equation}
which can be shown to be the generating functional of 1PI diagrams.
In particular, the fact that the effective action is defined as a
Legendre transformation implies 
\begin{align}
  &&\frac{\delta^R\Gamma}{\delta\tilde{\Phi}^A}=-J_A[\tilde{\Phi},\tilde\Phi^*],
  &&
     \frac{\delta^R\Gamma}{\delta\Phi^*_A}
     =\frac{\delta^R\mathcal{W}}{\delta\tilde\Phi^*_A}|_{J=J[\tilde\Phi,\tilde\Phi^*]} .
\end{align}
In terms of the effective action, the ST identity then takes the form
of the generalized {\bf Zinn Justin equation},
i.e., the master equation for the effective action,
\begin{equation}
  \frac{1}{2}(\Gamma,\Gamma)_{\tilde\Phi,\tilde\Phi^*}=0.
\end{equation}

\section{Elements of renormalization}
\label{sec:wess-zumino-cons}

The effective action is crucial for the renormalized theory: once its
is well-defined so is the complete theory. This is because connected
Green's functions can be constructed by using the effective action in
order to derive the Fyenman rules while summing over connected tree
graphs alone. The passage from connected to all Green's functions is
purely combinatorial, while the passage from Green's functions to
$S$-matrix elements through reduction formulas does not involve
ultraviolet issues either.

Gauge invariance is encoded in the master equation for the effective
action, whose derivation above was formal. In order to go beyond such
formal considerations, one may start by assuming the existence of a
regularisation scheme that is consistent with gauge invariance, such
as dimensional regularization in the absence of chiral fermions and
the Levi-Civita (pseudo-)tensor. In this case, the regularized
effective action $\Gamma_{\rm reg}$ continues to satisfy the master
action,
\begin{equation}
\frac{1}{2} (\Gamma_{\rm reg},\Gamma_{\rm reg})_{\tilde\Phi,\tilde\Phi^*}=0.
\end{equation}
When using that 
\begin{equation}
\Gamma_{\rm
  reg}[\tilde\Phi,\tilde\Phi^*]=S_\Psi[\tilde\Phi,\tilde\Phi^*]
+\hbar\Gamma^{(1)}+O(\hbar^2),\label{eq:51}
\end{equation}
with $\hbar$ keeping track of the loop order in the perturbative
expansion, and taking into account the master equation for $S_\Psi$,
the master equation for $\Gamma_{\rm reg}$ reduces at order $\hbar$ to
\begin{equation}
  \label{eq:52}
  (S_\Psi,\Gamma^{(1)})_{\tilde\Phi,\tilde\Phi^*}=0. 
\end{equation}
The 1-loop contribution is made of a divergent
and a finite part,  
\begin{equation}
\Gamma^{(1)}=\frac{1}{\e}\Gamma^{(1)}_{\rm div}+\Gamma^{(1)}_{\rm fin}, 
\end{equation}
where we implicitly are assuming dimensional regularization but
$1/\e$ could be replaced by another ultraviolet regulator
regulator. A crucial property is that the one-loop divergences are
{\bf local functionals}. We then get 
\begin{align}
(S_\Psi,\Gamma^{(1)}_{\rm fin})=0, && (S_\Psi,\Gamma^{(1)}_{\rm div})=0,
\end{align}
because the two terms are of different orders in $\e$. The latter
equation thus means that one-loop divergences are
BRST-closed
local functionals. If furthermore, they can be expressed as
\begin{equation}
\Gamma^{(1)}_{\rm div}=(S_\Psi,\Xi),
\end{equation}
for some local function $\Xi$ of ghost number -1, then they are
trivial divergences in the sense that they can be absorbed by a canonical
redefinition of fields and antifields, while the absorption of nontrivial
divergences requires the presence of appropriate coupling constants in
the starting point action.

Non trivial divergences are thus decribed by antifield-dependent BRST
cohomology in ghost number $0$ in the space of local functionals, 
\begin{equation}
H^0(s_\Psi)\simeq H^0(s),
\end{equation}
where we stress that this cohomology is isomorphic to the gauge
fixing independent cohomology associated to the minimal solution of
the master action since both are related by an
anticanonical transformation and the trivial pairs of the gauge fixing
sector drop out of cohomology.  

In renormalizable theories one can show that, once lower order
divergences have been absorbed to a given order, the divergences at
the next order are again BRST closed.

\begin{remark}
  The classical problem of constructing interactions consistent with
  gauge invariance (sometimes also refered to as ``Noether
  procedure'')
  can be formulated as a classical deformation problem using the BV
  formalism \cite{Barnich:1993vg}. Similarily to divergences which
  consitute a quantum deformation, infinitesimal deformations in that
  context are also controlled by antifield-dependent BRST cohomology
  in ghost number $0$ in the space of local functionals. The
  deformation parameter is no longer $\hbar$ but can for instance be
  related to the number of fields involved in the interactions or to
  the couplings constants. In the context of non-commutative
  Yang-Mills models, the existence of Seiberg-witten maps can also be
  discussed in these terms \cite{Barnich:2001mc,Barnich:2003wq}.
\end{remark}

Up to now, we have assumed that the regularization used was compatible
with gauge invariance, so that the master equation was still satisfied
after regularization. It is however possible that no such regulator
exists: in such a case the generating functional $\Gamma$ no longer
satisfies the master equation, but rather
\begin{equation}
\frac{1}{2}(\Gamma,\Gamma)=\hbar\mathcal{A}\circ\Gamma,
\end{equation}
where the insertion $\mathcal{A}$ corresponds again to a local
functional, called the \textbf{anomaly}. In this equation, $\Gamma$
either refers to the renormalized effective action in the case of of
power counting renormalizable theories and the equation is derived
using the so-called quantum action principle in the framework of BPHZ
renormalization (see \cite{Piguet:1995er} for a review), or it refers
to the regularized effective action for instance in the context of
dimensional renormalization, where so-called evanescent terms
proportional to $\epsilon$ and due to the regularization, produce the
right hand side (see \cite{Aoyama:1980yw,Tonin:1992wf} for more
details and also \cite{Barnich:2000me} for related
considerations). Note that in both these frameworks divergences,
respectively the associated finite amiguities in the BPHZ approach,
are still controlled by $H^0(s)$, even in the presence of anomalies.

The generalized \textbf{Wess-Zumino consistency
  conditions}
then follows:
\begin{align}
(\Gamma,(\Gamma,\Gamma))=0 && \implies && (\Gamma,\mathcal{A}\circ \Gamma)=0.
\end{align}
To the lowest order in $\hbar$, this means that
\begin{equation}
(S_\Psi,\mathcal{A})=0.
\end{equation}
In this case, if the anomaly is exact, 
\begin{equation}
\mathcal{A}=(S_\Psi,\Sigma),
\end{equation}
it can be reabsorbed by adding a local, BRST breaking counterterm,
$\Sigma$, $S_\Psi\to S_\Psi-\Sigma$. Hence, to lowest nontrivial
order, anomalies are characterized by antifield-dependent BRST
cohomology in ghost number $1$ in the space of local functionals,
\begin{equation}
H^1(s_\Psi)\simeq H^1(s),
\end{equation}
the determination of which is a gauge-invariant, computable problem.

In Yang-Mills theory, the \textbf{Adler-Bardeen anomaly} constitues a
famous example of such a cohomology class. It is computed in the
so-called universal algebra and corresponds to $\omega^{1,4}$ in the
chain of descent equations that relates a characteristic class in
form degree $6$ to
a primitive element in ghost number $3$,
\begin{equation}
  \begin{split}
    {\rm Tr}\ F^3&=d\omega^{0,5},\\
    s\omega^{0,5}+d_H\omega^{1,4}&=0,\\
    s\omega^{1,4}+d_H\omega^{2,3}&=0,\\
    \vdots &\\
    s\omega^{4,1}+d_H{\rm Tr}\ C^3&=0,\\
    s {\rm Tr}\ C^3&=0.
  \end{split}
\end{equation}

When imposing power-counting restrictions, semi-simple Yang-Mills
theories can be shown to be renormalizable quite easily in this
framework. Indeed, in this case, it is almost straightforward to show
that the only BRST cohomology class in ghost number $0$ in the space
of local functionals corresponds to the classical, gauge invariant
Yang-Mills action itself: all divergences which are not absorbable by
canonical field-antifield redefinitions can thus be absorbed by
redefinitions of the couplings associated with the different simple
factor of the gauge algebra.

The problem of completely characterizing
antifield-dependent
BRST cohomology in various ghost numbers, independently
of power-counting restrictions, has been addressed for Yang-Mills type
theories in \cite{Barnich:1994db,Barnich:1995ap} and reviewed in
\cite{Barnich:2000zw}. The inclusion of free abelian vector fields has
been completed more recently
\cite{Barnich:2017nty,Barnich:2018nqa}. Einstein gravity has been
considered in \cite{Barnich:1994mt}.

Based on these results, it has been argued \cite{Gomis:1996jp} that
these theories
are renormalizable in the modern sense, that is to say, in the sense
of effective field theories: for semi-simple Yang-Mills theory or
Einstein gravity, the antifield-dependent BRST cohomology is exhausted
by integrated, on-shell gauge invariant observables. Once these are
included with independent couplings from the very beginning, all
divergences can be absorbed, either by anticanonical field-antifield
redefinitions or by redefinitions of these couplings.

\begin{center}
  ----------
\end{center}

\paragraph{Summary} In this set of lectures, we have tried to set up
the relevant background material to compute antifield-dependent BRST
cohomology classes in the space of local functionals for generic
irreducible gauge theories. What physical information they provide may
be summarized in the following table:

\begin{center}
\begin{tabular}{ c |c }
  g 		& $H^g(s)$ 		\\  \hline
  \vdots	& $\emptyset$; 	\\  
  -3 	& $\emptyset$; 	\\
  -2		&	reducibility parameters (Killing vectors)
       / ADM-type surface charges\\
           -1 	&	conserved currents and global symmetries \\
  0		& 	loop divergences, consistent deformations \\
  1		&	anomalies \\
  \vdots	&

\end{tabular}
\end{center}

\section*{Acknowledgements}
\label{sec:acknowledgements}

\addcontentsline{toc}{section}{Acknowledgments}

The work of G.B.~is supported by the F.R.S.-FNRS Belgium,
convention FRFC PDR T.1025.14 and convention IISN 4.4503.15.




\providecommand{\href}[2]{#2}\begingroup\raggedright\endgroup

\end{document}